\documentclass[pra, twocolumn,  superscriptaddress, groupedaddress]{revtex4}
\usepackage{amssymb}
\usepackage{color}

\usepackage{amsthm}
\usepackage{graphicx}
\usepackage{dcolumn}
\usepackage{bm}
\usepackage{subfigure}
\usepackage{amssymb}
\usepackage{verbatim}
\usepackage{amscd}
\usepackage{amssymb}
\usepackage{amsmath, amsfonts}
\usepackage{setspace}
\usepackage[]{graphicx}        
\usepackage{amsthm}
\usepackage{enumerate}

\usepackage{bbm}

\newcommand{\tb}[1]{\textbf{#1}}

\newcommand{\assign}{:=}
\newcommand{\nocomma}{}

\newcommand{\tmop}[1]{\ensuremath{\operatorname{#1}}}
\newcommand{\tmtextbf}[1]{{\bfseries{#1}}}

\newtheorem{corollary}{Corollary}
\newtheorem{definition}{Definition}
\newtheorem{lemma}{Lemma}
\newtheorem{theorem}{Theorem}
\newtheorem*{theorem*}{Theorem}
\newtheorem*{lemma*}{Lemma}
\newtheorem{example}{Example}

\begin{document}

\title{Fault-tolerant logical gates in quantum error-correcting
codes}
\author{Fernando \surname{Pastawski}}
\affiliation{Institute for Quantum Information and Matter, California Institute of Technology, Pasadena, California 91125, USA
}
\author{Beni \surname{Yoshida}}
\affiliation{Institute for Quantum Information and Matter, California Institute of Technology, Pasadena, California 91125, USA
}

\begin{abstract}
Recently, Bravyi and K{\"o}nig have shown that there is a tradeoff between fault-tolerantly implementable logical gates and geometric locality of stabilizer codes. 
They consider locality-preserving operations which are implemented by a constant depth geometrically local circuit and are thus fault-tolerant by construction.
In particular, they shown that, for local stabilizer codes in $D$ spatial dimensions, locality preserving gates are restricted to a set of unitary gates known as the $D$-th level of the Clifford hierarchy. 
In this paper, we elaborate this idea and provide several extensions and applications of their characterization in various directions. 

First, we present a new no-go theorem for self-correcting quantum memory. 
Namely, we prove that a three-dimensional stabilizer Hamiltonian with a locality-preserving implementation of a non-Clifford gate cannot have a macroscopic energy barrier. 
This result implies that in Haah's Cubic code and Michnicki's welded code non-Clifford gates do not admit such an implementation. 

Second, we prove that the code distance of a $D$-dimensional local stabilizer code with non-trivial locality-preserving $m$-th level Clifford logical gate is upper bounded by $O(L^{D+1-m})$. 
For codes with non-Clifford gates ($m>2$), this improves the previous best bound by Bravyi and Terhal. 
Bombin's topological color codes saturate the bound for $m=D$. 

Third we prove that a qubit loss threshold of codes with non-trivial transversal $m$-th level Clifford logical gate is upper bounded by $1/m$. 
As such, no family of fault-tolerant codes with transversal gates in increasing level of the Clifford hierarchy may exist. 
This result applies to arbitrary stabilizer and subsystem codes, and is not restricted to geometrically-local codes. 

Fourth we extend the result of Bravyi and K{\"o}nig to subsystem codes. A technical difficulty is that, unlike stabilizer codes, the so-called union lemma does not apply to subsystem codes. This problem is avoided by assuming the presence of error threshold in a subsystem code, and the same conclusion as Bravyi-K{\"o}nig is recovered.
\end{abstract}


\maketitle

\section{Introduction}

Quantum error-correcting codes constitute an indispensable ingredient in the roadmap to fault-tolerant quantum computation as they offer the framework of enabling imperfect quantum gates and resources to implement arbitrarily reliable quantum computation~\cite{Shor96, Preskill98}. An essential feature for such codes is to admit a fault-tolerant implementation of a universal gate-set where physical errors should propagate in a benign and controlled manner. A paragon for fault-tolerant implementation of logical gates is provided by transversal unitary operations, \emph{i.e.} single qubit rotations acting independently on each  physical qubit. 

However, Eastin and Knill have proved that the set of transversal gates constitutes a finite group, and hence is not universal for quantum computation~\cite{Eastin09}, suggesting a tension between computational power and fault-tolerance. Recently, Bravyi and K{\"o}nig have further sharpened this tension for topological stabilizer codes supported on a lattice with geometrically local generators~\cite{Bravyi13b}. By extending their consideration to logical gates implemented by constant depth local quantum circuits as feasible proxy, they have shown that, in $D$ spatial dimensions, fault-tolerantly implementable logical gates are restricted to a set of unitary gates, known as the $D$-th level of the Clifford hierarchy~\cite{Gottesman99}. This result establishes a connection between two seemingly unrelated notions; fault-tolerance and geometric locality. 

The result by Bravyi and K{\"o}nig (BK) is motivated by considerations of topological stabilizer codes, which are also likely to suggest a host of future generalizations. In this paper, we begin to address open questions posed by the work of Bravyi and K{\"o}nig.

\subsection{Clifford hierarchy}

As in BK ~\cite{Bravyi13b}, the tensor product Pauli operators on $n$ qubits (denoted by $\mathsf{Pauli} = \langle X_j, Y_j, Z_j \rangle_{j \in [ 1, n]}$) and the corresponding Clifford hierarchy~\cite{Gottesman99} will play a central role. We provide a formal definition for the $m$-th level of the Clifford hierarchy $\mathcal{P}_m$.

\begin{definition}\label{def:CliffordHierarchy}
  We define the \tmtextbf{Clifford hierarchy} as $\mathcal{P}_0 \equiv  \mathbbm{C}$ (i.e. global complex phases), and then recursively as
  \begin{equation}
    \mathcal{P_{}}_{m + 1} = \{ U : \forall P \in \mathsf{Pauli},\  U P U^\dagger P^\dagger
    \in \mathcal{P}_m \}.
  \end{equation}
\end{definition}

Note that despite using a commutator in place of conjugation, the above definition coincides with the usual one for $m \geq 2$~\cite{Gottesman99, Bravyi13b}. (See appendix~\ref{sec:comparison} for comparison). $\mathcal{P}_{1}$ is a group of Pauli operators with global complex phases. $\mathcal{P}_{2}$ coincides with the \emph{Clifford group} and includes the Hadamard gate $H$, $\pi/2$ phase shift and the $\tmop{CNOT}$ gate. $\mathcal{P}_{3}$ includes some non-Clifford gates such as $\pi/4$ phase shift and the Toffoli gate. $\pi/2^{m-1}$ phase shift belongs to $\mathcal{P}_{m}$. Note that $\mathcal{P}_{m}$ is a set and is not a group for $m\geq 3$.

The Gottesman-Knill theorem assures that any quantum circuit composed exclusively from Clifford gates in $\mathcal{P}_2$, with computational basis preparation and measurement, may be efficiently simulated by a classical computer~\cite{Nielsen_Chuang}. In contrast, incorporating any additional non-Clifford gate to $\mathcal{P}_2$ results in a universal gate set. In theory, gates in the Clifford group can be implemented with arbitrarily high precision by using concatenated stabilizer codes~\cite{Gottesman98} or topological codes. Realistic systems also offer decoherence-free implementation of some Clifford gates. For instance, braiding of Ising anyons, that are believed to exist in the fractional quantum Hall effect state at filling fraction $\nu=5/2$, implements certain Clifford gates with an estimated error-rate being $10^{-30}$~\cite{Bravyi05}. For this reason, it is important to fault-tolerantly perform \emph{non-Clifford} logical gates outside of $\mathcal{P}_2$. 

\subsection{Summary of results}

Let us now summarize the main contributions of this work. We begin by providing a self-contained and arguably simpler derivation of BK's result. 
We then derive a key technical lemma to assess fault-tolerant implementability of logical gates for both stabilizer and subsystem error-correcting codes (lemma~\ref{lemma:hierarchy} in section~\ref{sec:review}). 

In addition, there are four main original contributions. Below, we provide a preliminary statement of each, deferring a more rigorous treatment to later sections.

\subsubsection{No-go result for self-correction}

First of all, we show that the property of self-correction imposes a further restriction on logical gates implementable by constant depth local circuits. Namely, we find that the assumption of having no string-like logical operators reduces the level of the implementable Clifford hierarchy by one with respect to BK's result. 

\begin{theorem*}\emph{\tb{[Self-correction]}}
If a $D$-dimensional stabilizer Hamiltonian, consisting of geometrically local terms with bounded norms, has a macroscopic energy barrier, the set of logical gates, admitting a locality-preserving implementation, is restricted to $\mathcal{P_{}}_{D - 1}$.
\end{theorem*}

This theorem allows us to obtain a new no-go result for self-correcting quantum memory in three spatial dimensions; a three-dimensional topological stabilizer Hamiltonian with a locality-preserving non-Clifford gate cannot have a macroscopic energy barrier. The proof is presented in section~\ref{sec:self-correction}. The result establishes a somewhat surprising connection between ground state properties and excitation energy landscape. While technically simple, this observation is arguably the most interesting. 

\subsubsection{Upper bound on code distance}

Our second result concerns a tradeoff between the code distance and locality-preserving implementability of logical gates. Namely, we find that implementability of logical gates from the higher-level Clifford hierarchy reduces an upper bound on the code distance of a topological stabilizer code. 

\begin{theorem*}\emph{\tb{[Code distance]}}
If a stabilizer code with geometrically-local generators in $D$ spatial dimensions admits a locality-preserving implementation of a logical gate $U \in \mathcal{P}_{m}$ for $m\geq2$ (but $U\not \in \mathcal{P_{}}_{m-1}$), then its code distance is upper bounded by $d \leq O(L^{D+1-m})$.
\end{theorem*}

For a code with a non-Clifford gate ($m>2$), this result improves the previous best bound $d \leq O(L^{D-1})$ for topological stabilizer codes~\cite{Bravyi09}. The bound is found to be tight for $m=D$ as Bombin's topological color codes saturates it~\cite{Bombin06,Bombin14}. The proof is presented in section~\ref{sec:self-correction}. The theorem also applies to a topological subsystem code if its stabilizer subgroup admits a complete set of geometrically local generators. Such subsystem codes include Bombin's topological \emph{gauge} color code~\cite{Bombin14}. 

\subsubsection{Loss threshold}

Our third result relates the loss threshold in stabilizer and subsystem error-correcting codes with the set of transversally implementable logical gates. 

\begin{theorem*}\emph{\tb{[Loss threshold]}}
  Suppose we have a family of subsystem codes with a loss tolerance $p_l > 1
  / n$ for some natural number $n$. Then, any transversally implementable
  logical gate must belong to $\mathcal{P}_{n -
  1}$.
\end{theorem*}

We would like to emphasize that the above theorem does \emph{not} assume geometric locality of generators or lattice structures, and holds for arbitrary stabilizer \emph{and} subsystem codes. The proof is presented in section~\ref{sec:loss}. 

\subsubsection{Subsystem code and the Clifford hierarchy}

Finally, the main technical result is to generalize BK's result to subsystem codes with local generators. A difficulty is that the so-called union lemma does not apply to a topological subsystem code~\cite{Bravyi10, Bravyi11}. Minimal supplementary assumptions, such as a finite loss threshold for the code and a logarithmically increasing code distance, are required in order to recover the same thesis as BK's for locality-preserving logical gates.

\begin{theorem*}\emph{\tb{[Subsystem code]}}
Consider a family of subsystem codes with geometrically local gauge generators in $D$ spatial dimensions such that the code has a constant loss threshold and a code distance growing at least logarithmically in the number of physical qubits. 
Then, any locality-preserving logical unitary, fully supported on an $m$-dimensional region ($m\leq D$), has a logical action included in $\mathcal{P}_{m}$.
\end{theorem*}

The proof is presented in section~\ref{sec:subsystem}. Supplementary assumptions arise from considerations on fault-tolerance of the code. A finite loss threshold is necessary for a finite error threshold against depolarization. A logarithmically increasing code distance is necessary for the recovery failure probability to vanish at least polynomially in the number of physical qubits. Supplementary assumptions are not required for subsystem codes with geometrically local stabilizer generators as the union lemma holds for such codes.

\subsection{Organization of the paper}

The paper is organized as follows. In section~\ref{sec:review}, we provide a definition of subsystem codes and derive a key technical tool to study fault-tolerant implementability of logical gates. We then provide a derivation of BK's result. In section~\ref{sec:loss}, we derive a tradeoff between the loss-tolerance and transversal implementability of logical gates. In section~\ref{sec:subsystem}, we generalize BK's result to topological subsystem codes. 
In section~\ref{sec:self-correction}, we find a restriction on the set of logical gates admitting a locality-preserving implementation arising from self-correction. 
We then derive an upper bound on the code distance of topological stabilizer codes. 
Section~\ref{sec:discussion} is devoted to summary and discussion.

\section{Fault-tolerance versus locality}\label{sec:review}

In this section, we review the framework of subsystem error-correcting codes and derive a tool relating fault-tolerant implementability of logical gates and locality (or non-locality) of logical gates in multi-partitions. We also present a qualitative derivation of BK's result for topological stabilizer codes.

\subsection{Fault-tolerant implementation of logical gates}

Let us begin with a brief review of the stabilizer formalism~\cite{Gottesman96}. Given the Hilbert space of $n$ qubits $\mathcal{H} = ( \mathbbm{C}^2)^{\otimes n}$, a Pauli stabilizer group $\mathcal{S}$ is an abelian subgroup of the Pauli group on $n$ qubits which does not contain $-\mathbbm{1}$. The codeword space of the stabilizer group $\mathcal{S}$ is defined to be the subspace $\mathcal{C ( \mathcal{S}) \subseteq \mathcal{H}}$ of common $+1$ eigenvectors for stabilizers in $\mathcal{S}$ (\emph{i.e.} $\mathcal{C} ( \mathcal{S}) = \left\{ | \psi \rangle \in \mathcal{H} : \forall S \in \mathcal{S} , \ S | \psi \rangle = | \psi \rangle \right\}$). \emph{Topological} stabilizer codes are characterized by having their constituent physical qubits laid out on a $D$-dimensional lattice, in such a way that the stabilizer group $\mathcal{S}$ admits a complete set of geometrically local generators $\mathcal{S} = \langle S_1, \ldots, {S_{}}_{n-k} \rangle$ (\emph{i.e.} each generator $S_j$ is supported on a ball of constant radius $\xi$). Here $k$ is the number of logical qubits encoded in the codeword space $\mathcal{C} ( \mathcal{S})$ when $S_{j}$ are independent generators. In the present paper, the word \emph{topological} refers to quantum error-correcting codes defined on a lattice with geometrically local generators. 

Ideally, one hopes a logical gate $U$ to be implemented by a transversal unitary operation (\emph{i.e.} an operator with a tensor product form $U=\otimes_{j=1}^{n}U_{j}$ where $U_{j}$ is a single qubit rotation acting on $j$-th physical qubit) so that local errors at physical qubits do not propagate to other qubits. 
For fault-tolerant implementation, it is desirable that a logical gate $U$ admits an implementation by a constant-depth quantum circuit in order to avoid uncontrolled error propagations. 
Ideally, the gates in such a circuit should additionally be geometrically local to simplify their physical realization.
For this reason, it is important to classify logical gates of quantum error-correcting codes admiting such an implementation. 
In the present paper, the word \emph{locality-preserving} refers to a logical unitary operator that can be implemented by a constant-depth geometrically-local circuit. 
Note that stabilizer codes admit transversal implementation of Pauli logical gates in $\mathcal{P}_{1}$, and CSS stabilizer codes admit quasi-transverse
\footnote{If two copies of a CSS code are stacked such that corresponding qubits are geometrically close, performing pairwise $\tmop{CNOT}$ on all physical qubits implements a $\tmop{CNOT}$ gates on all pairs of encoded qubits in the two copies.} 
implementations of certain $\tmop{CNOT}$ gate in $\mathcal{P}_{2}$.

Bravyi and K{\"o}nig˜\cite{Bravyi13b} consider the set of logical gates that may be realized on a topological stabilizer code with a constant depth local quantum circuit. Their main result is stated as follows:

\begin{theorem}\emph{\tb{[Bravyi and K{\"o}nig]}}
If $U$ is a morphism between D-dimensional topological stabilizer codes $\mathcal{C}_1$ and $\mathcal{C}_2$, and U is implementable by a constant-depth quantum circuit with short-range gates, then $U$ is a $\mathcal{P}_D$-morphism for all large enough $L$.\label{Thm:BK}
\end{theorem}

Note that the theorem by Bravyi and K{\"o}nig deals with code deformations~\cite{Bombin09c} where logical transformations between two different codes $\mathcal{C}_1$ and $\mathcal{C}_2$ are also considered. In the present paper, we do not deal with code deformations for simplicity of discussion by assuming $\mathcal{C}_1 =\mathcal{C}_2$. Our arguments may be made applicable to the case for $\mathcal{C}_1\not=\mathcal{C}_2$ with a little effort.

\subsection{Gauge and logical qubits}

One important aim of the present paper is to generalize BK's result to topological subsystem codes. In this work, we will refer to subsystem codes to denote the Pauli stabilizer formalism, which provide a generalization of stabilizer QECCs to the context of operator quantum error correction formalism~\cite{Poulin05, Kribs06, Nielsen07}. Intuitively, a subsystem code is a stabilizer code defined by $\mathcal{S}$ where we encode quantum information into only a subset of the qubits in the stabilized subspace. Encoded qubits in this subset will be called logical qubit and the remaining qubits will be called gauge qubits (\emph{i.e.} the stabilized subspace may be decomposed into $\mathcal{H}_{\tmop{logical}} \otimes \mathcal{H}_{\tmop{gauge}} = \mathcal{C} ( \mathcal{S})$ as in Fig.~\ref{fig_subsystem}). 

A subsystem code is concisely defined by its gauge group $\mathcal{G} \subseteq \mathsf{Pauli}$ which may be non-abelian and contain $- 1$ (in contrast to the stabilizer group $\mathcal{S}$). 
The stabilizer subgroup $\mathcal{S}$ consists of centers of the gauge group $\mathcal{G}$ (i.e. elements of $\mathcal{G}$ that commute with all the elements in $\mathcal{G}$), and is defined as $\mathcal{S} = Z ( \mathcal{G})/\mathbbm{C}$, where signs are consistently chosen for the operators in the center $Z ( \mathcal{G})= \{ z \in \mathcal{G} : \forall g_{} \in \mathcal{G}, z g = g z \}$ such that $-1$ is not included in the group.
This leaves some freedom for $\mathcal{S}$ associated to the signs of its generators. 
The codespace of a subsystem code, denoted by $\mathcal{C}(\mathcal{S})$, is stabilized by $\mathcal{S}$, and logical qubits are encoded in a subsystem where gauge operators act trivially. The case $\mathcal{S} = \mathcal{G}$ corresponds to the special case of stabilizer codes.

One merit of subsystem codes is that the error recovery procedure may admit simpler realizations with measurements on fewer-body Pauli operators since it is not necessary to worry about errors affecting gauge qubits~\cite{Bacon06}. As such, one might expect that imposing locality on such codes could be less restrictive than doing so on stabilizer codes in terms of transversally implementable logical gates. However, the present work suggests that there is no significant advantage for subsystem codes.

\begin{figure}[htb!]
\centering
\includegraphics[width=0.85\linewidth]{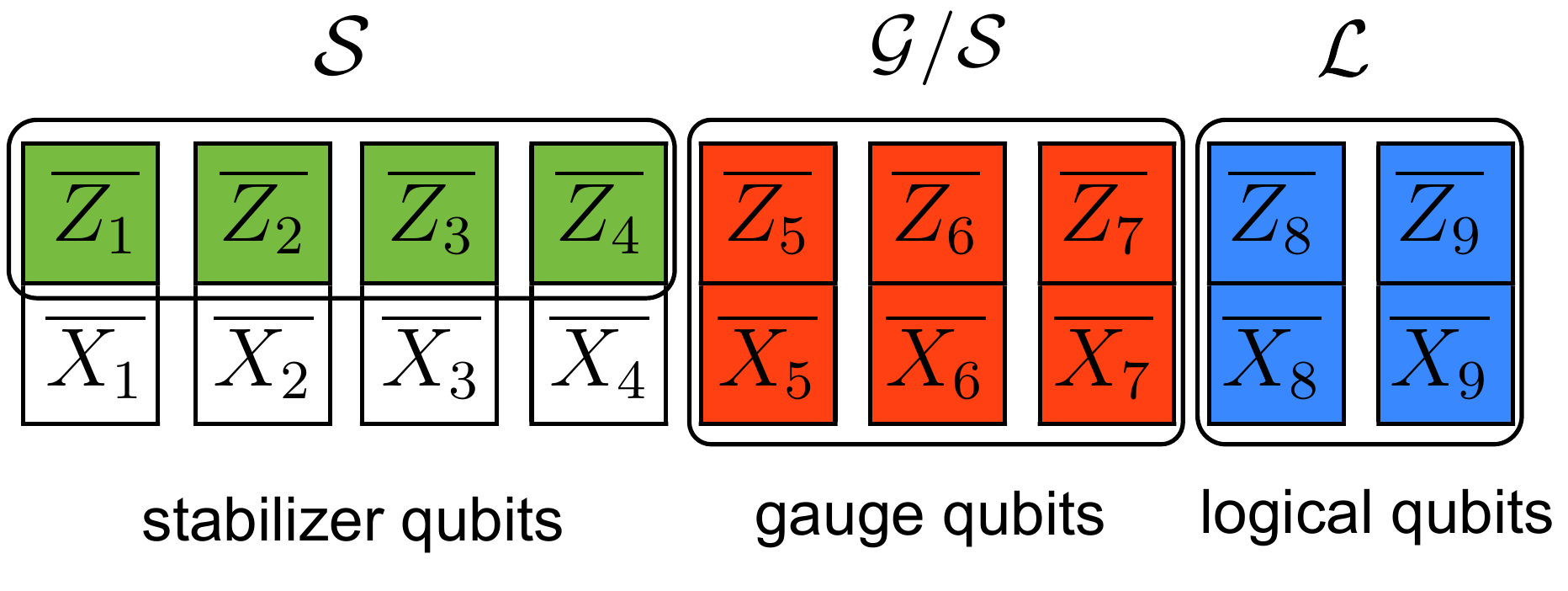}
\caption{The algebraic structure of the gauge group $\mathcal{G}$ and the stabilizer subgroup $\mathcal{S}$ in a subsystem code. The figure illustrates an example with $n=9$ qubits, three gauge qubits and two logical qubits. $\mathcal{G}/\mathcal{S}$ is the full Pauli algebra for the gauge qubits (red online). The stabilizer group $\mathcal{S}$ is generated by the $Z$ operators on the stabilized qubits (green online). 
The remaining qubits (blue online) represent the algebra for logical qubits $\mathcal{L}$. Generators of each group are brought to a canonical form via an appropriate Clifford transformation $U$ such that $\overline{X_{j}}=UX_{j}U^{\dagger}$ and $\overline{Z_{j}}=UZ_{j}U^{\dagger}$ for $j=1,\ldots,n$.
} 
\label{fig_subsystem}
\end{figure}

\subsection{Bare and dressed logical operators}

Logical operators preserve the codespace $\mathcal{C} ( \mathcal{S})$ and act non-trivially on logical qubits. In a subsystem code, there are two types of logical operators, called \emph{bare} and \emph{dressed} logical operators, depending on how they act on gauge qubits~\cite{Bravyi11}. Given a decomposition of the codespace as $\mathcal{C} ( \mathcal{S})=\mathcal{H}_{\tmop{logical}} \otimes \mathcal{H}_{\tmop{gauge}}$, bare logical operators act exclusively on logical qubits and act trivially on gauge qubits: $[U_{\tmop{bare}}]=[U]_{L} \otimes [I]_{G}$ where $[U]_{L}$ represents a logical action of $U_{\tmop{bare}}$ on logical qubits, and $[I]_{G}$ represents a trivial action on gauge qubits. Formally, bare Pauli logical operators are the centralizers of the gauge group $\mathcal{G}$ (i.e. Pauli operators that commute with all the elements of $\mathcal{G}$): ${\mathcal L}_{\tmop{bare}} = C ( \mathcal{G})= \{ z \in \textsf{Pauli}: \forall g \in \mathcal{G}, z g = g z \}$ denotes the centralizer of $\mathcal{A}$. Bare logical operators are identified up to stabilizer operators ${\mathcal S}$ since stabilizers act trivially both on gauge and logical qubits. 

The centralizer group $C ( \mathcal{G})$ consists only of bare Pauli logical operators in $\mathcal{P}_{1}$. Bare logical operators, beyond the Pauli group, are defined as follows:

\begin{definition}
A unitary operator $U$ is a bare logical operator of the subsystem code defined by the gauge group $\mathcal{G}$ (with associated stabilizer subgroup $\mathcal{S}$ and projector onto the code space $P_{\mathcal{C (\mathcal{S})}}$) iff
  \begin{equation}
    [ U, P_{\mathcal{C ( \mathcal{S})}}] = 0 \hspace{1em} \tmop{and}
    \hspace{1em}  [ U, G]
    P_{\mathcal{C ( \mathcal{S})}} = 0 \hspace{1em} \forall G \in \mathcal{G}.
  \end{equation}
\end{definition}

Dressed logical operators may act both on logical and gauge qubits non-trivially. Formally, dressed Pauli logical operators are the centralizer of $\mathcal{S}$: ${\mathcal L}_{\tmop{dressed}} = C( {\mathcal S})$: and its logical action has a form of $[U_{\tmop{dressed}}] = [U]_{L} \otimes [U]_{G}$. Dressed logical operators are identified up to gauge operators in $\mathcal{G}$ since gauge operators act trivially on logical qubits. Note that dressed Pauli operators decomposes into unitary actions on the logical and gauge qubits and have tensor product structures. 

A caution is needed in dealing with dressed logical operators beyond the Pauli group. We say an operator is a dressed logical unitary when its action decomposes into the logical and gauge qubits with tensor product structures. Note that, there exist unitary operators that preserve the stabilized subspace $C ( \mathcal{S})$, but do not have tensor product structures over $\mathcal{H}_{\tmop{logical}} \otimes \mathcal{H}_{\tmop{gauge}}$. We exclude such unitary operators since the unitary action on the logical qubits is dependent on the state of the gauge qubits. The premise of a subsystem code is that one does not worry about errors acting on gauge qubits, and thus such operators cannot be used to transform encoded logical states in a controlled way. 

Formally, dressed logical operators, the ones with tensor product structures in  $\mathcal{C} ( \mathcal{S})=\mathcal{H}_{\tmop{logical}} \otimes \mathcal{H}_{\tmop{gauge}}$, are defined as follows:

\begin{definition}
An operator $U$ is a dressed logical unitary on a subsystem code defined by the gauge group $\mathcal{G}$ and associated stabilizer subgroup $\mathcal{S}$ if and only if 
  \begin{equation}
  \frac{1}{| \mathcal{G} |} \sum_{G \in \mathcal{G}} G U \rho U^{\dagger}
  G^{\dagger} = \frac{1}{| \mathcal{G} |} \sum_{G \in \mathcal{G}} U G \rho
  G^{\dagger} U^{\dagger} .
\end{equation}
for all $\rho = P_{\mathcal{C ( \mathcal{S})}} \rho P_{\mathcal{C ( \mathcal{S})}}$.
\end{definition}

In other words, conjugation by $U$ commutes with depolarization with respect to the gauge group. Intuitively, the definition imposes that tracing out gauge qubits does not affect the action on logical qubits. 

These definitions of bare and dressed logical operators beyond the Pauli group are indeed algebraically well defined. For instance, a set of all the dressed logical operators form a closed group under multiplication. Furthermore, the set of bare logical operators is preserved under conjugation by dressed logical operators:

\begin{lemma}\label{lemma_closed}
Let $U_{d}$ be a dressed logical operator and $U_{b}$ be a bare logical operator for a subsystem code. Then $U_{d}U_{b}U_{d}^{\dagger}$ is also a bare logical operator.
\end{lemma}

\begin{proof}
Each of $U_{d}$ and $U_{b}$ preserves the codespace $P_{\mathcal{C ( \mathcal{S})}}$ and so does their product. We will now prove that $U_{d}U_{b}U_{d}^{\dagger}$ commutes with any gauge operator $G_0 \in \mathcal{G}$ restricted to $P_{\mathcal{C ( \mathcal{S})}}$.
  \begin{eqnarray}
    U_d U_b U_d^{\dagger} G_0 P_{\mathcal{C ( \mathcal{S})}} & = & U_d
    \frac{1}{| \mathcal{G} |} \sum_{G \in \mathcal{G}} G G^{\dagger} U_b
    P_{\mathcal{C ( \mathcal{S})}} U_d^{\dagger} G_0 \nonumber\\
    & = & \frac{1}{| \mathcal{G} |} \sum_{G \in \mathcal{G}} U_d G U_b
    P_{\mathcal{C ( \mathcal{S})}} G^{\dagger} U_d^{\dagger} G_0  \nonumber\\
    & = & \frac{1}{| \mathcal{G} |} \sum_{G \in \mathcal{G}} G U_d U_b
    P_{\mathcal{C ( \mathcal{S})}} U_d^{\dagger} G^{\dagger} G_0  \nonumber\\
    & = & \frac{1}{| \mathcal{G} |} \sum_{G \in \mathcal{G}} G_0 G U_d U_b
    P_{\mathcal{C ( \mathcal{S})}} U_d^{\dagger} G^{\dagger} \nonumber
  \end{eqnarray}
We have multiplied by an identity, used the commutation of $P_{\mathcal{C (\mathcal{S})}}$ with all the operators. Then we use the definition of $U_d$ as a dressed logical operators. Finally we use that $G_0 G \in \mathcal{G}$ to relabel the sum. We conclude by reverting the previous steps, leaving $G_0$ as the leftmost factor.
\end{proof}

This lemma allows us to formally prove that dressed logical operators may transform logical states of logical qubits in a way independent of logical states of gauge qubits. 

\begin{lemma}\label{lemma_tensor}
Let $|\psi\rangle,|\psi'\rangle$ be an arbitrary pair of states in the codespace $\mathcal{P}_{C(\mathcal{S})}$ such that $\langle \psi | U_{b} | \psi \rangle = \langle \psi' | U_{b} | \psi' \rangle$ for all the bare logical operators $U_{b}$. Then, for any dressed logical operator $U_{d}$, one has $\langle \psi | U_{d}U_{b}U_{d}^{\dagger} | \psi \rangle = \langle \psi' | U_{d}U_{b}U_{d}^{\dagger} | \psi' \rangle$ for all $U_{b}$.
\end{lemma}

\begin{proof}
Lemma~\ref{lemma_closed} implies that $U_{b}'=U_{d}U_{b}U_{d}^{\dagger}$ is a bare logical operator. Then
  \begin{align}
    \langle \psi | U_{d}U_{b}U_{d}^{\dagger} | \psi \rangle 
    =& \langle \psi | U_{b}' | \psi \rangle\nonumber \\
    =& \langle \psi' | U_{b}' | \psi' \rangle\nonumber \\
    =& \langle \psi' | U_{d}U_{b}U_{d}^{\dagger} | \psi' \rangle. \nonumber
  \end{align}\label{lemma:tensor}
\end{proof}

\subsection{Cleaning lemma}

The notion of \emph{cleaning}, initially introduced for stabilizer codes~\cite{Bravyi09} also arises for subsystem codes. Let us begin by reviewing the cleaning procedure for stabilizer codes. Consider a logical Pauli operator $P \in \mathcal{L}$ that has non-trivial supports on some subset $R$ of qubits. Formally, the cleaning of a logical operator $P$ from the subset $R$ refers to a procedure of multiplying a logical Pauli operator $P$ by an operator $S \in \mathcal{S}$ to obtain an equivalent logical operator $PS$ that has a trivial action on $R$. The cleaning is not always possible, and can be performed if and only if there exists a stabilizer $S$ whose action on $R$ is identical to the action of $P$ on $R$: \emph{i.e.} $P|_{R}=S|_{R}$ up to a complex phase where $P|_{R},S|_{R}$ represent restrictions of $P,S$ onto a subset $R$.

The cleaning lemma by Bravyi and Terhal states that, if a subset $R$ supports no logical operator (except the one with trivial action), then any logical operator $P$ can be cleaned from $R$~\cite{Bravyi09}. Namely, there exists a stabilizer $S$ such that $PS$ is supported exclusively on qubits in the complementary subset $R^{c}$. A result in~\cite{Beni10} concisely relates the set of independent logical operators supported on a pair of complementary subsets of qubits. In particular, for any subset $R$ of qubits, one may define $l( R)$ to be the number of independent Pauli logical operators supported exclusively on $R$.

\begin{lemma}
Suppose a stabilizer code has $k$ logical qubits. Then $l ( R) + l( R^c) = 2 k$.\label{lemma:partition_stab}
\end{lemma}

The cleaning lemma is recovered from this lemma by imposing that there is no logical operator supported on $R$: $l ( R)=0$ which leads to $l( R^c) = 2 k$. Since there are $2k$ independent Pauli logical operators for a stabilizer code with $k$ logical qubits, all the logical operators have representations that are supported exclusively on $R^c$. Thus the cleaning from the subset $R$ is always possible.

In the case of subsystem codes, multiplication by an element of $\mathcal{S}$ preserves bare logical operators, whereas multiplication by an element of $\mathcal{G}$ preserves dressed logical operators. A result due to Bravyi~\cite{Bravyi11} generalizes lemma~\ref{lemma:partition_stab} to relate the set of independent bare and dressed logical operators supported on complementary regions. In particular, we may define $l_{\tmop{dressed}}( R)$ and $l_{\tmop{bare}} ( R)$ to be the number of independent dressed and bare Pauli logical operators supported on $R$.

\begin{lemma}
Suppose a subsystem code has $k$ logical qubits. Then $ l_{\tmop{dressed}} ( R) + l_{\tmop{bare}} ( R^c) = 2 k$.\label{lemma:partition_sub}
\end{lemma}

This lemma implies that, if there are no non-trivial dressed (bare) logical operators fully supported on $R$, all the bare (dressed) logical operators can be cleaned from a region $R$ so that they are supported exclusively on $R^{c}$. This leads to the following definition for bare (dressed)-cleanable regions.

\begin{definition}
  A region $R$ is \tmtextbf{bare (dressed)-cleanable}, if it supports no non-trivial dressed (bare) logical operators.\label{def:BareCleanable1}
\end{definition}

Cleanability is closely related to coding properties of the code. The code distance $d$ for a subsystem code is defined as the size of the smallest possible support for an operator in $\mathcal{L}_{\tmop{dressed}} \setminus \mathcal{G}$ (a dressed Pauli logical operator having non-trivial action on $\mathcal{H}_{\tmop{logical}}$). Furthermore, a subset $R$ of qubits is \emph{correctable} if and only if it supports no dressed logical operator. In other words, a subset $R$ is correctable if and only if $R$ is bare-cleanable. 

\subsection{Fault-tolerant logical gate and cleanability}

Let us now present a key technical lemma which plays a central role in deriving all the main results in the present paper.

\begin{lemma}\label{lemma:hierarchy}
Let $\{ R_j \}_{j \in [ 0, m]}$ be a set of regions where $R_0$ is bare-cleanable and each of the regions $\{ R_j \}_{j \in [ 1, m]}$ is dressed-cleanable in a subsystem code. If a dressed logical unitary $U$ is supported on the union $\bigcup_{j \in [ 0, m]} R_{j}$ and is transversal with respect to regions $R_{j}$, then the logical action $[U]_{L}$ of $U$ on the logical qubits correspond to an element of $\mathcal{P}_m$ (the $m$-th level of the Clifford hierarchy). 
\end{lemma}

The above theorem does not require geometric locality of the gauge or stabilizer generators, and thus applies to arbitrary subsystem codes, nor does it require that the support to be the full set of qubits.

\begin{proof}
The proof proceeds by induction. Assuming $m = 0$, the operator $U$ is fully supported on a bare-cleanable region $R_0$. All the bare logical Pauli operators may be supported on ${R_0}^c$ hence they must commute with $U$. Thus, $[ U]_L$ must be a trivial logical operator in $\mathcal{P}_0$ (proportional to identity).
  
Let us now prove the inductive step. We assume that all the dressed transversal operators supported on the union $R_{0} \cup \bigcup_{j=1}^{m} R_{j}$ are in $\mathcal{P}_m$. Consider a transversal dressed logical operator $U$ such that
  \begin{equation}
    \tmop{supp} ( U) \subseteq R_0 \cup \bigcup_{j = 1}^{m + 1} R_j.
  \end{equation}
By definition, the logical action of $U$ has a tensor product form $[ U] = [ U]_L \otimes [ U]_G$ where $[U]_{L}, [U]_G$ denote the logical actions on the logical and gauge qubits respectively. Since $R_{m+1}$ is dressed-cleanable, all the dressed Pauli operators may be supported on ${R_{m+1}}^{c}$, and their logical actions $[P] = [P]_L \otimes [P]_G$ have a tensor product form due to lemma~\ref{lemma_tensor}. Hence, the group commutator $U P U^{\dagger}P^{\dagger}$ is also a dressed logical operator with a tensor product form with respect to the gauge and logical qubits. 
Furthermore, transversality of $U$ and $P$ with respect to $R_{j}$ mandates
  \begin{equation}
\tmop{supp} ( U P U^{\dagger} P^{\dagger}) \subseteq R_0 \cup \bigcup_{j = 1}^{m } R_j 
  \end{equation}
which implies $[ U P U^{\dagger}P^{\dagger}]_L \in \mathcal{P}_m$. By definition of the Clifford hierarchy, $[ U]_L \in \mathcal{P_{}}_{m+1}$.
\end{proof}

\section{Loss-tolerance and transversal logical gates}\label{sec:loss}

One implication that may be obtained at this point is a tradeoff between particle loss threshold and the set of achievable transversal gates. Indeed, we will see that increasing the first comes at the expense of restricting the second.

A highly desirable property for quantum error-correcting codes is that they must, with high probability, tolerate errors (such as depolarization) on a constant fraction $p_e$ of randomly chosen physical qubits. Here, $p_e$ is called an error threshold for a family of codes if,\ the probability of correcting independent and identically distributed errors, occurring with probability $p < p_e$, approaches to unity for members of the family with increasingly large number of physical qubits. 

An important form of errors is erasure errors which correspond to loss of physics qubits from the system. In addition to the fact that loss errors are unavoidable in realistic physical systems, the loss-tolerance is necessary for quantum error-correcting codes to have a finite error threshold. Namely, any form of depolarizing noise is more severe than qubit loss since, in the latter, full information on the location of errors is available. 
\footnote{Formally, erasure errors are modeled by extending the space associated to each physical particle with one additional state $| l \rangle$ which indicates loss of the corresponding particle. An error-correcting recovery map for loss errors may mimic the one for depolarizing noise by mapping all particles marked as $| l \rangle$ to the fixed-point of corresponding depolarizing channel. Hence, the loss threshold $p_l$ must necessarily be no smaller than any depolarization threshold $p_l \geq p_e$. }

The following corollary elucidates the existing tension between loss-threshold and the set of transversally implementable gates.

\begin{theorem}
Suppose we have a family of subsystem codes with a loss tolerance $p_l > 1 / n$ for some natural number $n$. Then, any transversally implementable logical gate must belong to $\mathcal{P_{}}_{n -  1}$.\label{Thm:SubsystemLossThreshold}
\end{theorem}

\begin{proof}
Suppose $p_l > 1 / n$, and assign each qubit to one of $n$ regions $\{ R_j \}_{j \in [ 0, n-1]}$ uniformly at random. Each of the regions chosen this way will be correctable with a probability which is arbitrarily close to unity as we take larger codes from the family. Finally, we may conclude  by applying lemma~\ref{lemma:hierarchy}  to the $n$ correctable regions obtained in this way, which are both bare and dressed cleanable.
\end{proof}

The above result applies to arbitrary stabilizer and subsystem codes, and is \emph{not} restricted to codes with geometrically local generators. 

\begin{example}\emph{
The toric code saturates the bound of theorem~\ref{Thm:SubsystemLossThreshold} in that it has a loss threshold of $p_l = 1 / 2 > 1 / 3$ and can still transversely implement some logical operators in $\mathcal{P}_2$ (such as $\tmop{CNOT}$)~\cite{Stace09}.
  }
\end{example}

\begin{example}\emph{
Reed-Muller code $[ [ 2^m - 1 , 1, 3]]$ admits transversal implementation of $\pi/2^{m-1}$ phase shift which belongs to $\mathcal{P}_m$~\cite{Nielsen_Chuang}. As a family of codes with increasing $m$, it must have a zero loss threshold.
  }
\end{example}

\begin{example}\emph{
$D$-dimensional topological color code admits transversal implementation of gates in $\mathcal{P}_D$ but not of gates in ${\mathcal{P}_{}}_{D + 1}$. Its loss threshold is hence upper bounded by $1 / D$. This conclusion may likely be recovered by other arguments related to percolation in $D$-dimensional lattices.
  }
\end{example}

\section{Constant depth circuits and geometric locality}\label{sec:subsystem}

Discussions so far do not rely on geometric locality of required generators in the code, which is one of the most important features to assess its experimental feasibility. The underlying assumption of geometric locality is that physical qubits are associated to particles on a regular lattice and check operators involve only particles within a constant sized neighborhood. More precisely, the gauge group $\mathcal{G}$ may be generated by a set of Pauli operators, each one having support restricted to a ball of diameter $\xi = O ( 1)$. In this section, we generalize BK's result to topological subsystem codes that are supported on a $D$-dimensional lattice with geometrically local generators.

\subsection{Union lemma} 

A challenge in generalizing BK's result is that the so-called \emph{union lemma} does apply to topological subsystem codes. The union lemma for a topological stabilizer code states that \emph{the union of two spatially disjoint cleanable regions is also cleanable}. Here two regions are spatially disjoint if local stabilizer generators overlap with at most one of the regions. 

\begin{lemma}\emph{\tb{[Union lemma (stabilizer code)]}}
\label{lemma_union_stab}
For a topological stabilizer code, let $R_1$ and $R_2$ be two spatially disjoint regions such that there exists a complete set of stabilizer group generators $\{ S_{j} \}$ each intersecting at most one of $\{ R_{1}, R_{2} \}$. If $R_1$ and $R_2$ are cleanable, then the union $R_1 \cup R_2$ is also cleanable.
\end{lemma}

At this point, let us review the derivation of BK's result in order to illustrate the use of the union lemma. For a topological stabilizer code with a growing code distance, one is able to split the $D$-dimensional space into $D+1$ regions $R_{m}$ for $m=0,\ldots,D$ where $R_{m}$ consists of small regions with constant size connected components which are spatially disjoint. Let us demonstrate it for $D=2$ (see Fig.~\ref{fig_union}). We first split the entire lattice into patches of square tiles so that the diameter of local stabilizer generators is much shorter than the spacing of tiles. This square tiling has three geometric object; points, lines and faces. First, we ``fatten'' points to create regions $R_{0}$. We then fatten lines and create regions $R_{1}$. The remaining regions are identified to be $R_{2}$. Therefore $R_{m}$ is the union of fattened $m$-dimensional objects. For a $D$-dimensional lattice, we start with a $D$-dimensional hyper-cubic tiling and fatten $m$-dimensional objects to obtain $R_{m}$ for $m=0,\ldots,D$.

Each of connected components in $R_{m}$ is cleanable as the code distance is growing with the system size $n$. Also connected components in $R_{m}$ are spatially disjoint. Due to the union lemma, the union of spatially disjoint small regions is correctable, and thus $R_{m}$ is correctable. Then lemma~\ref{lemma:hierarchy} implies that transversally implementable logical gates are restricted to $\mathcal{P}_{D}$, recovering BK's result (Theorem~\ref{Thm:BK}). 

\begin{figure}[htb!]
\centering
\includegraphics[width=0.60\linewidth]{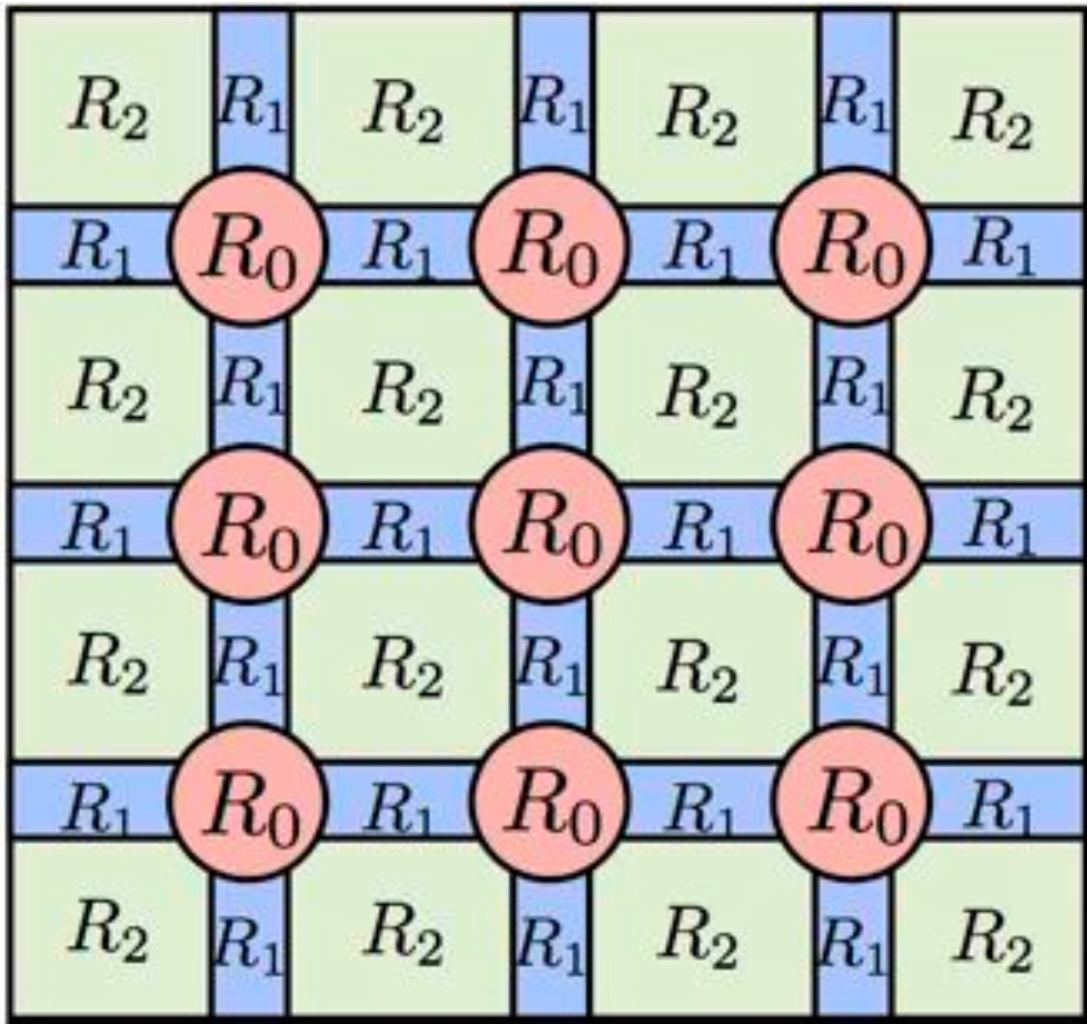}
\caption{The partition of a two-dimensional lattice into three regions $R_{0},R_{1},R_{2}$ which consist of smaller regions that are correctable and spatially disjoint. 
} 
\label{fig_union}
\end{figure}

For a topological subsystem code, two regions are said to be spatially disjoint if local gauge generators may overlap with at most one of the regions. Unlike a topological stabilizer code, however, geometric locality of stabilizer generators is not always guaranteed since the stabilizer subgroup $\mathcal{S}$ is defined to be the center of the gauge group $\mathcal{G}$, and generators of $\mathcal{S}$ are,  in general, products of multiple local gauge generators. 
As such, the union lemma holds only for dressed-cleanable regions as summarized below.

\begin{lemma}\emph{\tb{[Union lemma (subsystem code)]}}
\label{lemma_union_sub}
For a topological subsystem code, let $R_1$ and $R_2$ be two spatially disjoint regions such that there exists a complete set of gauge group generators $\{ G_{j} \}$ each intersecting at most one of $\{ R_{1}, R_{2} \}$. If $R_1$ and $R_2$ are dressed-cleanable, then the union $R_1 \cup R_2$ is also dressed-cleanable.
\end{lemma}

It is worth emphasizing that the union lemma for bare-cleanable regions are recovered for a topological subsystem code if its stabilizer subgroup admits a complete set of geometrically local generators. This is the case for Bombin's gauge color code is a three-dimensional subsystem code. 

\subsection{Fault-tolerance and non-local stabilizer generators}

In addition to the technical difficulty, the breakdown of the union lemma seems to taint fault-tolerance of a subsystem code. Emergence of geometrically non-local stabilizer generators prevents us from having the union lemma for bare-cleanable regions. Indeed, this is the case for two and three-dimensional quantum compass models~\cite{Dorier05, Bacon06}. We should yet mention that geometrically non-local stabilizer generators are hard to measure reliably and hence undesirable for physical realizations. Namely, when non-local stabilizer generators are supported by a large number of physical qubits, their measurements cannot be performed fault-tolerantly. 

Macroscopic code distance $d$ is no doubt a necessary requirement for a family of codes to provide reliable error-correcting properties. Formally, it imposes that the code distance $d$ may be made arbitrarily large by increasing the code block size $n$. Macroscopic code distance $d$ is necessary for the probability of failing the error-correcting procedures to be small. Namely, the failure probability is lower bounded by $p_{\tmop{fail}}\geq O(p^{d})$ where $p$ is the physical error rate. In order for $p_{\tmop{fail}}$ to converge to zero for sufficiently small but finite $p$, the distance $d$ must be macroscopic. 
It is highly desirable to have a code with exponentially vanishing $p_{\tmop{fail}}$. 
Polynomially decaying $p_{\tmop{fail}}$ may be still acceptable, which imposes that the distance $d$ must grow at least logarithmically. 
However, sub-polynomially vanishing $p_{\tmop{fail}}$ would be impractical. 
As such we may assume that a code distance grows at least logarithmically in a fault-tolerant code.

Macroscopic code distance is necessary, but not sufficient for a finite error threshold in the code. For a stabilizer code with low-weight generators, a macroscopic code distance guarantees a finite error threshold against depolarization. Namely, if the code distance grows logarithmically with $n$, a finite error threshold is guaranteed as proven by Kovalev and Pryadko~\cite{Kovalev13}.

\begin{theorem}\emph{\tb{[Error threshold  \cite{Kovalev13}]}}
Consider a family of stabilizer codes whose stabilizer group generators have constant weight. If the code distance grows logarithmically in the system size $n$, a finite error threshold always exists such that the recovery failure probability $p_{fail}$ approaches to zero as $n$ increases. 
\end{theorem}

Yet this theorem does not apply to all subsystem codes, and applies only to subsystem codes with low-weight stabilizer generators. Indeed, two and three-dimensional quantum compass models have a zero error threshold, and thus are not scalable quantum error-correcting codes~\cite{Pastawski09}. To our knowledge, there is no known relation between the existence (or absence) of a finite error threshold and locality (of non-locality) of stabilizer generators. It seems plausible that such a relation could exist.

BK's derivation relies on a macroscopic code distance, which is required for a finite error threshold. In the present work, we use the fault-tolerance itself as the guiding principle. Namely, we will assume that (i) the code distance grows at least logarithmically, and (ii) the code has a finite (loss) error threshold.

\subsection{Bravyi-K{\"o}nig for subsystem code}

Let us now proceed to generalization of BK's result to a topological subsystem code. The distance $d(,)$ between particles on the lattice will be used to define the $r$-neighborhood $\mathfrak{B} ( R, r)$ of a region $R$ which includes region $R$ and all particles within distance $r$ to it. Furthermore, we define the spread $s_U$ of a unitary as the smallest possible distance such that $\forall A: \tmop{supp}(UAU^\dagger) \subseteq \mathfrak{B} ( \tmop{supp}(A), s_U)$. In particular, if $U$ is implemented by a constant depth circuits composed of geometrically local gates, the spread $s_U$ will also be bounded by a constant.

A version of lemma~\ref{lemma:hierarchy} involving the lattice geometry can now be stated.

\begin{lemma}\label{lemma:hierarchy2}
Let $U$ be a dressed logical unitary operator supported on the union of mutually non-intersecting regions $R_{0}$ and $\{  R_j \}_{j \in [ 1, m]}$. If $R_{0}$ is bare-cleanable and each $R_j^+ \assign\mathfrak{B} ( R_j, 2^{j-1}s_U)$ is dressed-cleanable for $j>0$, then the logical unitary implemented by $U$ belongs to $\mathcal{P}_{m}$.
\end{lemma}

This means that when dealing with locality-preserving circuits which implement logical unitary gates, it is sufficient to use extended correctable regions such that they overlap in a boundary of width $2^{m - 1} s_U$, where $m$ is the number of regions to be used. 
As such, much of discussion dealing with transversal gates applies to finite depth circuits. 
The proof is presented in appendix~\ref{sec:proof}.

With an assumption of macroscopic code distance alone, one is able to obtain the following statement for topological subsystem codes.

\begin{corollary}\label{Cor:SusbsystemSimpleBound}
Consider a family of subsystem codes with increasing code distance defined by geometrically local gauge generators of diameter bounded by $\xi$ in $D$ spatial dimensions. Then the set of dressed logical unitary gates implementable by constant depth circuits is included in $\mathcal{P_{}}_{D + 1}$.
\end{corollary}

\begin{proof}
Since gauge generators are geometrically local with diameter bounded by $\xi$, the union lemma (lemma~\ref{lemma_union_sub}) applies to dressed cleanable regions that are separated by a distance $\xi$ or larger. Furthermore, any region with volume smaller than $d$ is dressed-cleanable by the definition of a code distance. Let $s_{U}$ be the spread of the circuit $U$. One has $d > ( 2^D s_U + \xi)^D$ for sufficiently large $n$ since the code has a macroscopic distance. Then the lattice may be partitioned into $D+1$ disjoint regions $\{ R_j \}_{j \in [ 1, D + 1]}$ such that $R_j^+ \assign \mathfrak{B} ( R_j, 2^{j - 1} s_U)$ is dressed-cleanable for all $j>0$. For instance, we construct a $D$-dimensional hyper-cubic tiling and fatten $m$-dimensional objects to obtain ${R_{}}_{m+1}$ for $m=0,\ldots,D$. By taking $R_{0}$ to be an empty set $\emptyset$, we conclude that the logical action of $U$ is included in ${P_{}}_{D + 1}$.
\end{proof}

Note that $R_{j}$ are dressed-cleanable, but not necessarily bare-cleanable since the union lemma does not hold for bare-cleanable regions. Thus, we needed to take $R_{0}$ to be an empty set, which result in increasing the level of the implementable Clifford hierarchy by one with respect to BK's result for topological stabilizer codes. An interesting open problem is to find subsystem codes with growing distance which achieve the bound stated in corollary \ref{Cor:SusbsystemSimpleBound}. If such subsystem codes exist, we believe that they would be highly artificial and would possess highly non-local stabilizer generators. 

We now further assume that the family of codes has a non-zero loss threshold $p_l > 0$ and that a code distance $d$ grows at least logarithmically with the number of particles $n$. Under these reasonable and perhaps indispensable assumptions for fault-tolerance of the code, we obtain the same thesis as BK's result for topological subsystem codes.

\begin{theorem}
Consider a family of subsystem codes with geometrically local gauge generators in $D$ spatial dimension with i) a loss threshold $p_l>0$ and ii) a code distance $d = \Omega ( \log^{1-1/D}( n) )$. Then any dressed logical unitary that can be implemented by a constant depth geometrically local circuit $U$ must belong to $\mathcal{P}_D$.
\end{theorem}

As a side note, we remark that our proof technique borrows an idea by Hastings which was used on a different topic~\cite{Hastings11}.

\begin{proof}
Let us assume for simplicity that $U$ is transversal. The argument leading to lemma~\ref{lemma:hierarchy2} suffices to make the current proof applicable to a constant depth geometrically local circuit by taking care of some cumbersome yet inessential caveats. 

Imagine that some subset of qubits, denoted as $R_{\tmop{loss}}$, is lost. This subset $R_{\tmop{loss}}$ is chosen so that each site has an independent probability $p_{0}<p_{l}$ of being included in $R_{\tmop{loss}}$. By definition of loss error threshold, $R_{\tmop{loss}}$ must be correctable (in other words, bare-cleanable) with probability approaching to unity as the system size $n$ grows. The key idea is to make use of this randomly generated bare-cleanable region $R_{\tmop{loss}}$ to construct a bare-cleanable region $R_0$ which consists of spatially disjoint balls of constant radius.

For any fixed region $R$, the probability that $R$ is included in $R_{\tmop{loss}}$ is given by $\Pr ( R \subseteq R_{\tmop{loss}}) = p_0^{| R |}$. So, given a ball of radius $r \gg \xi $, it is included in $R_{\tmop{loss}}$ with some constant probability independent of $n$. Let us now split the full lattice into unit cells of volume $v_c = c \log ( n)$ as in Fig.~\ref{fig_random}. Inside a given unit cell, the probability of having no ball of radius $r$ included in $R_{\tmop{loss}}$ is $O ( 1 / \tmop{poly} ( n))$ where the power of $n$ can be made arbitrary large by increasing a finite constant $c$. Hence, with probability approaching to unity, $R_{\tmop{loss}}$ includes at least one ball of radius $r$ in each unit cell. We choose one ball from each unit cell so that they are spatially disjoint, and denote its union as $R_{0}$. Then a bare-correctable region $R_{0}$ consists of balls of diameter $r$ that are spatially disjoint with at most $O(\log(n)^{1/D})$ linear separation. Imagine a skewed $D$-dimensional hyper-cubic tiling by drawing lines which connect balls in $R_{0}$ (see Fig.~\ref{fig_random}). We then fatten $m$-dimensional objects to construct a covering of the full lattice with $R_{m}$ for $m=0,\ldots,D$.

It remains to prove that $R_{m}$ for $m>0$ are dressed-cleanable. Any region with volume smaller than $d = \Omega ( \log^{1-1/D}(n))$ is cleanable. For $m<D$, $R_{m}$ consists of connected components with volume at most $O( \log^{1-1/D} (n))$, and hence are dressed-cleanable. For $R_{D}$, suppose that there exists a non-cleanable $D$-dimensional connected component, denoted as $R$, with volume $O(\log(n))$. Then $R$ must support at least one bare logical Pauli operator $U_{\tmop{bare}}$. Yet, the disentangling lemma~\cite{Bravyi10} tells that $U_{\tmop{bare}}$ can be supported by qubits that live on the boundary of $R$, whose volume is at most $ O( \log^{1-1/D} (n))$, leading to a contradiction. Therefore, $R_{D}$ is dressed-cleanable. Given a bare-cleanable region $R_{0}$ and dressed cleanable regions $R_{m}$ for $m=1,\ldots,D$, lemma~\ref{lemma:hierarchy} implies that transversally implementable $U$ must be included in $\mathcal{P}_{D}$.
\end{proof}

\begin{figure}[htb!]
\centering
\includegraphics[width=0.70\linewidth]{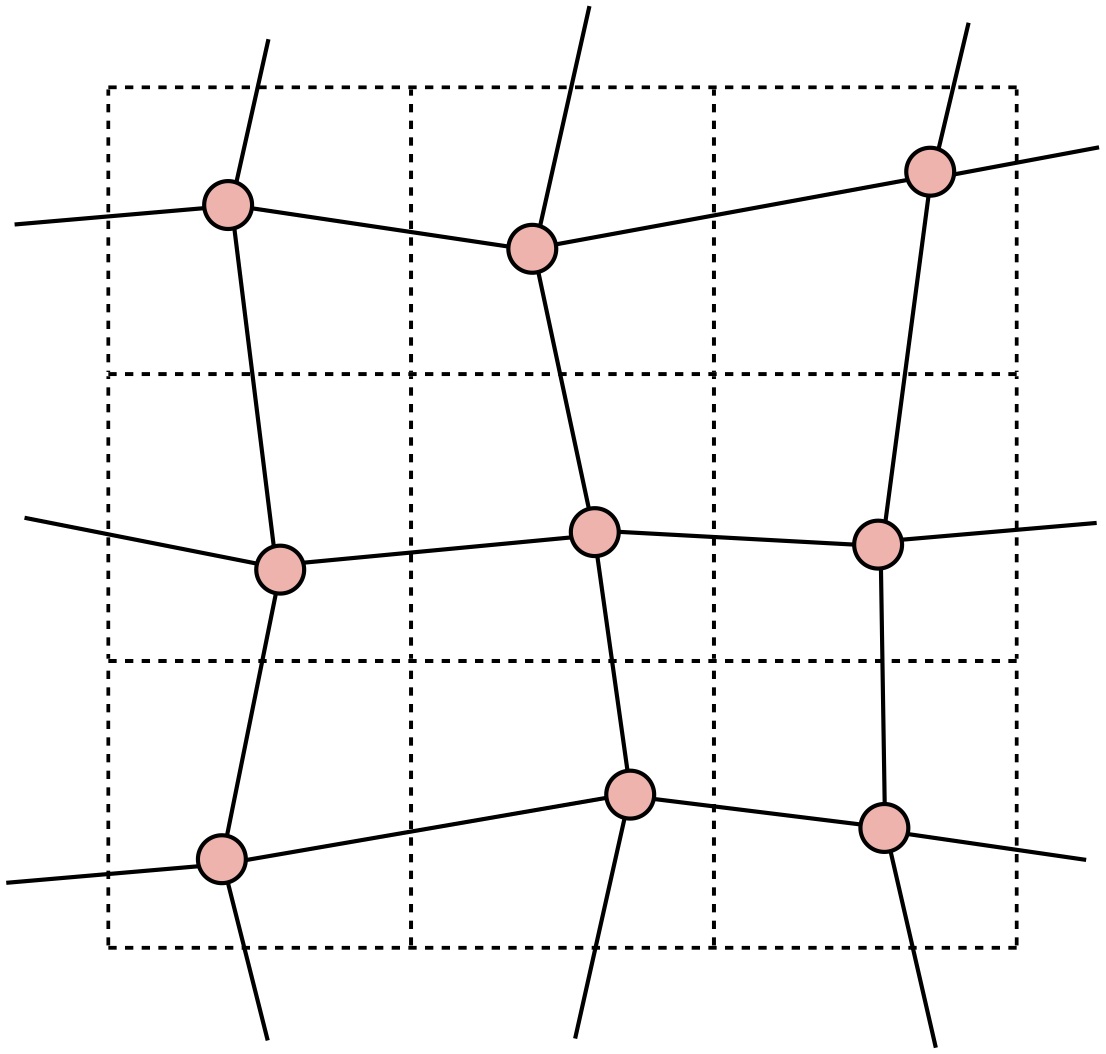}
\caption{A construction of a bare-cleanable region $R_{0}$. Circles represent balls that are included in randomly generated subset $R_{\tmop{loss}}$ of qubits. Dotted lines mark unit cells with volume $O(\log(n))$.
} 
\label{fig_random}
\end{figure}

A further observation is that constant depth circuits supported on a string-like region must be Pauli operators, and in general, constant depth logical operators supported on a $m$-dimensional region must be in $\mathcal{P}_m$ regardless of the spatial dimension of the lattice $D\geq m$. 

\section{Non-clifford gate prohibits self-correction}\label{sec:self-correction}

The problem of self-correcting quantum memories seeks to provide a Hamiltonian where the energy landscape prevents qubit errors at the physical level from accumulating and irreversibly introducing a logical error in contact with a thermal environment~\cite{Dennis02, Bravyi09}. Formally, self-correcting quantum memory is defined as a many-body quantum system where a logical qubit may be encoded for a macroscopic time~\cite{Beni11}. An important question is whether such a system may exist in three spatial dimensions. No-go results have ruled out most two-dimensional systems and a certain class of three-dimensional systems~\cite{Bravyi09,Haah10, Beni11,Landon-Cardinal13}, and no known three-dimensional model has macroscopic quantum memory time.  

In this section, we derive a new no-go result on three dimensional self-correcting quantum memory that arises from fault-tolerant implementability of a non-Clifford gate. 
In particular, we show that a stabilizer Hamiltonian with a locality-preserving non-Clifford gate cannot have a macroscopic energy barrier, and thus it is not expected to provide a practical increase in memory time in terms of the system size $n$. 
We then derive an upper bound on the code distance of topological stabilizer codes with locality-preserving logical gates from the higher-level Clifford hierarchy. 

\subsection{Self-correction and fault-tolerance}

For a topological stabilizer code, the stabilizer Hamiltonian is composed of geometrically local operators in the stabilizer group: $H = - \sum_{j} S_{j}$ where $S_{j}\in \mathcal{S}$. A non-rigorous yet commonly used proxy to assess whether self-correction can be achieved is the presence of a macroscopic energy barrier that scale with the system size. Macroscopic energy barrier seems to be a necessary yet insufficient condition for the system to exhibit macroscopic memory time~\footnote{Models proposed in~\cite{Haah11, Michnicki12} have a macroscopic energy barrier, yet quantum memory time is upper bounded by constant, which is perhaps due to topological transition temperature being zero. Finite transition temperature is not sufficient to guarantee exponentially growing quantum memory time~\cite{Beni14}}. For stabilizer Hamiltonians, the presence of string-like logical operators implies the absence of a macroscopic energy barrier~\footnote{The absence of string-like logical operators does not necessarily imply the presence of macroscopic energy barrier~\cite{Bravyi11b}.}. 

Here, we find a tradeoff on locality-preserving logical gates arising from a macroscopic energy barrier in a stabilizer Hamiltonian. 

\begin{theorem}
If a stabilizer Hamiltonian in $D$ spatial dimensions has a macroscopic energy barrier, the set of fault-tolerant logical gates is restricted to $\mathcal{P}_{D - 1}$.\label{Thm:NoStringStabilizer}
\end{theorem}

\begin{proof}
Let $R_0, R_1, \ldots, {R_{}}_{D-1}$ be regions which jointly cover the whole lattice. Each region is a collection of disjoint parallel tubes with a fixed orientation (see Fig. \ref{fig:ParallelTubes}). This covering can generically be achieved for a $D$-dimensional lattice. The presence of a macroscopic energy barrier implies the absence of string-like logical operators. Since there are no logical operators supported on individual tubes, there are no logical operators supported on any of single regions $R_j$ due to the union lemma. In other words, regions $R_j$ are cleanable. Applying lemma~\ref{lemma:hierarchy2}, we conclude that constant depth logical gates should be restricted to $\mathcal{P_{}}_{D - 1}$.
\end{proof}

\begin{figure}[htb!]
\centering
\includegraphics[width=0.85\linewidth]{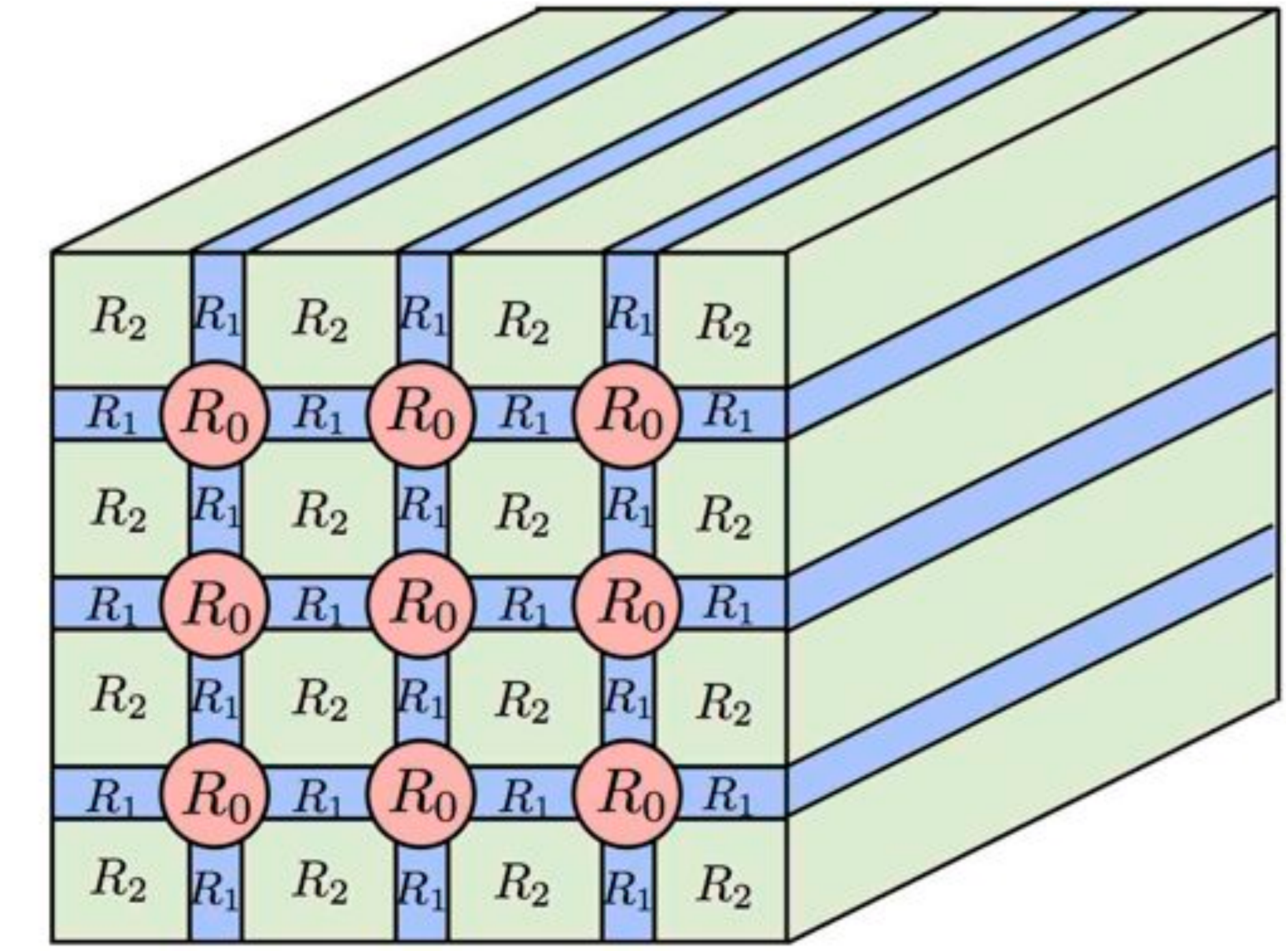}
\caption{The partition of the lattice into $R_{0},R_{1},\ldots,{R_{}}_{D-1}$ for $D=3$. 
} 
\label{fig:ParallelTubes}
\end{figure}

Haah~\cite{Haah11, Bravyi13} provided the first example of a three-dimensional topological stabilizer code which is free from of string-like logical operators. The code is defined on a three dimensional $L \times L \times L$ cubic lattice with an energy barrier scaling as $O( \log L)$. There also exist a number of three-dimensional translation symmetric stabilizer codes which are free from string-like logical operators~\cite{Kim12, Beni13}. By theorem~\ref{Thm:NoStringStabilizer}, for $D=3$, the presence of a macroscopic energy barrier implies that the set of locality-preserving logical gates is restricted to $\mathcal{P}_2$.

\begin{corollary}
Haah's 3D stabilizer code~\cite{Haah11} has no constant depth logical gates outside of $\mathcal{P}_2$.
\end{corollary}

A different approach to construct stabilizer codes with a macroscopic energy barrier has been proposed by Michnicki~\cite{Michnicki12}, who introduced the notion of code welding to construct new codes by combining existing ones. The welding technique leads to a construction of a topological stabilizer code with a polynomially growing energy barrier in three spatial dimensions. Our theorem~\ref{Thm:NoStringStabilizer} also applies to this code.

\begin{corollary}
Michnicki's 3D welded stabilizer code has no constant depth logical gates outside of $\mathcal{P}_2$.
\end{corollary}

A model of a six-dimensional self-correcting quantum memory with fault-tolerantly implementable non-Clifford gates has been proposed~\cite{Bombin13}. An intriguing question is whether such a code may exist in four (or five) spatial dimensions or not.

We then move to discussion on topological subsystem codes. A generic recipe to construct Hamiltonians for topological subsystem codes is not known. A candidate Hamiltonian, often discussed in the literature, is composed of geometrically local terms in the gauge group: $H= - \sum_{j} G_{j}$~\footnote{Due to the non-commutativity of $G_{j}$, a ground state $|\psi\rangle$ of the Hamiltonian $H= \sum_{j} \alpha_j G_{j}$ is not necessarily inside the stabilized subspace. Indeed it is analytically and computationally difficult to find values of $S_{j}$ in the ground space. For CSS subsystem codes, the ground space of $H= - \sum_{j} G_{j}$ is guaranteed to be in the stabilized subspace defined by $\mathcal{S}$ due to the Perron-Frobenius theorem~\cite{Ocko11}.}. Regardless of the choice of the Hamiltonian, the presence of bare-logical operators with string-like support implies the absence of an energy barrier as long as terms in the Hamiltonian consist only of local generators of the gauge group $\mathcal{G}$. 

For topological subsystem codes, we obtain a less restrictive tradeoff between fault-tolerant implementability and geometric non-locality of logical gates.

\begin{corollary}
  If a topological subsystem code in $D$ spatial dimensions has macroscopic energy barrier, the set of transversal operators is restricted to $\mathcal{P_{}}_D$. 
  \label{Cor:NoStringSubsystem}
\end{corollary}

The three-dimensional gauge color code has transversal gates in $\mathcal{P}_2$ and do not have string-like bare logical operators, and hence are not ruled out from having a macroscopic energy barrier.

\subsection{Upper bound on code distance}

Geometric non-locality of logical operators, such as no-string hypothesis, imposes a restriction on fault-tolerantly implementable gates in topological stabilizer codes as in theorem~\ref{Thm:NoStringStabilizer}. Reversing the argument, one may observe that fault-tolerant implementability of logical gates from the higher-level Clifford hierarchy imposes a restriction on geometric non-locality of logical operators. 

Here we find a tradeoff between the code distance and fault-tolerant implementability of logical gates. 

\begin{theorem}
If a topological stabilizer code in $D$ spatial dimensions admits a locality preserving implementation for a logical gate from $\mathcal{P}_{m}$, but outside of $\mathcal{P}_{m-1}$, its code distance is upper bounded by $d \leq O(L^{D+1-m})$.\label{thm:Distance}
\end{theorem}

\begin{proof}
Let $R_0, R_1, \ldots, R_{m-1}$ be regions which jointly cover the whole lattice. Each region is a collection of disjoint parallel $D+1-m$-dimensional objects (Fig. \ref{fig:ParallelTubes} corresponds to the case for $D=3$ and $m=3$). Suppose that there is no logical operator supported on any of single regions $R_j$. Applying lemma~\ref{lemma:hierarchy}, implementable logical operators are restricted to $\mathcal{P_{}}_{m-1}$, leading to a contradiction. Thus, at least one region $R_{j}$ supports a logical operator. Due to the union lemma, such a logical operator can be supported on a single $D+1-m$-dimensional object whose volume is $O(L^{D+1-m})$, which implies $d \leq O(L^{D+1-m})$.
\end{proof}

Bravyi and Terhal have derived an upper bound on the code distance for topological stabilizer and subsystem codes: $d \leq O(L^{D-1})$~\cite{Bravyi09}. Whether the Bravyi-Terhal bound is tight for $D\geq3$ remains open. For $m=2$, our bound is reduced to the Bravyi-Terhal bound. Theorem~\ref{thm:Distance} implies that locality-preserving implementations of non-Clifford gates imposes a further restriction on the code distance of topological stabilizer codes.

Topological color code, proposed in a seminal work by Bombin~\cite{Bombin06, Bombin14}, is a $D$-dimensional topological stabilizer code that admits transversal logical gates from the $D$-th level of the Clifford hierarchy. The code has a string-like logical operator, and thus $d=O(L)$, implying that our bound is tight for $m=D$.

\begin{example}\emph{
  Bombin's $D$-dimensional topological color code saturates the bound in theorem~\ref{thm:Distance}.
  }
\end{example}

It would be interesting if this hypothesis could be combined with the code threshold hypothesis to strengthen the conclusion. 

\section{Conclusions}\label{sec:discussion}

We have provided several extensions of BK's characterization of fault-tolerantly implementable logical gates. Our results are summarized as follows: (i) A three-dimensional stabilizer Hamiltonian with a fault-tolerantly implementable non-Clifford gate is not self-correcting. (ii) The code distance of a $D$-dimensional topological stabilizer code with non-trivial $m$-th level logical gate is upper bounded by $O(L^{D+1-m})$. (iii) A loss threshold of a subsystem code with non-trivial $m$-th level transversal logical gate is upper bounded by $1/m$. (iv) Fault-tolerantly implementable logical gates in a $D$-dimensional topological subsystem code belong to the $D$-th level $\mathcal{P}_{D}$ in the presence of a finite error threshold. 

Open questions include the possibility of further generalizing the result of Bravyi and K{\"o}nig to other families of codes such as frustration-free commuting projector codes. In an upcoming article, we will present a Bravyi and K{\"o}nig type characterization of logical operations implementable by constant depth circuits in the context of topological quantum field theories. 

Another interesting direction to extend these results concerns topological codes with non-geometrically-local gates, and quantum LDPC codes. 
It has been recently proven by the authors that, for families of the toric code and color codes, local constant-depth gates (not necessarily geometrically-local) do not increase the level of the implementable Clifford hierarchy. 
Dissipative dynamics may also be utilized for fault-tolerant logical implementation of topological codes~\cite{Pastawski11}. 

The definition of quantum phases, widely accepted in the literature, is that, two ground state wavefunctions belong to different phases if there is no local unitary transformation connecting them~\cite{Chen10}. Yet even within the ground space of a Hamiltonian, it is possible that different ground states are in different ``phases''. Perhaps, Bravyi and K{\"o}nig type characterization will give a coherent insight on classification of ground state wavefunctions with long-range entanglement.

Fault-tolerant implementability of non-Clifford logical gates is an important ingredient for magic state distillation protocols~\cite{Bravyi05b}. 
An interesting future problem includes the asymptotic rate of the number of magic states that can be distilled with a desired precision. Finally, it may be interesting to study the gauge-fixing technique~\cite{Paetznick13, Bombin14} and code concatenation~\cite{OConnor14} from the viewpoint of Bravyi and K{\"o}nig type characterization.

\section*{Acknowledgements} 

We would like to thank Michael Beverland for pointing out non-geometric interpretation of BK's results, and Robert K{\"o}nig, David Poulin and John Preskill for fruitful discussions. We acknowledge funding provided by the Institute for Quantum Information and Matter, a NSF Physics Frontiers Center with support of the Gordon and Betty Moore Foundation (Grants No. PHY-0803371 and PHY-1125565). BY is supported by the David and Ellen Lee Postdoctoral fellowship.

\appendix

\section{Observations on the Clifford hierarchy}\label{sec:comparison}

In the present work, we have adopted a slightly different definition for the Clifford hierarchy $\mathcal{P}_n$ with respect to the one introduced by Gottesman and Chuang~\cite{Gottesman99} and used by Bravyi and K{\"o}nig
~\cite{Bravyi13b}. In this appendix, we would like to justify that they are mostly equivalent yet the alternate definition permits stating our results in a more compact manner. Let us recall the usual definition.

\begin{definition}
The Clifford hierarchy is usually defined as follows. The first level of the hierarchy is taken to be equivalent to the Pauli group $\mathcal{P}_1 \equiv \mathsf{Pauli}$. Successive levels of the hierarchy are defined recursively as\label{def:CliffordUsual}
\begin{equation}
    \mathsf{Clifford}_{m + 1} = \left\{ U : \forall P \in \mathsf{Pauli},
    \ \nocomma U P U^{\dagger} \subseteq \mathsf{Clifford}_m
    \right\} .
\end{equation}
\end{definition}

The following statement shows how our definition \ref{def:CliffordHierarchy}
is equivalent to definition \ref{def:CliffordUsual}.

\begin{lemma}
$\mathcal{P}_1 = \mathcal{\mathbbm{C} \cdot \mathsf{Pauli}}$ and $\mathcal{P}_n = \mathsf{Clifford}_{n}$ for $n \geq 2$.
\end{lemma}

\begin{proof}
  Let us first show that $\mathcal{P}_1 = \mathcal{\mathbbm{C} \cdot \mathsf{Pauli}_{}}$. Suppose that $U \in \mathcal{P}_1$. By hypothesis, the group commutator of $U$ with any Pauli operator $P \in \mathcal{P}$ is trivial up to a phase $U P U^{\dagger} P^{\dagger} = e^{i \theta}$. This phase must be $\pm 1$, since it is an eigenvalue for the rank one superoperator resulting from conjugation by a Pauli operator $P_{} U^{\dagger} P^{\dagger} = e^{i \theta} U^{\dagger}$. Conversely, we may consider the rank one superoperator $U \cdot U^{\dagger}$ for which the Pauli operators constitute a full set of eigenoperators with eigenvalues $\pm 1$. This uniquely determines $U$ to be equivalent to a Pauli operator itself up to a global phase. Here, we have crucially used that the Pauli operators linearly span the full operator algebra.
  
We will now prove that $U \in \mathcal{P}_n \Leftrightarrow U \in \mathsf{Clifford}_n$ by induction for $n \geq 2$. The proof relies on the observation that all the levels of the usual Clifford hierarchy are closed under right multiplication by Pauli operators $\mathcal{\mathsf{Clifford}}_n = \mathcal{\mathsf{Clifford}}_n \cdot \mathsf{Pauli}$ which can be proven inductively.
  
Suppose $U \in \mathcal{P}_{n + 1}$. Hence, \ for any $P \in \mathsf{Pauli}$ we \ have that $U P U^{\dagger} P^{\dagger} \in \mathcal{P}_{n}$ and consequently $U P U^{\dagger} \in \mathsf{Clifford}_n$. This implies that $U \in \mathsf{Clifford}_{n + 1}$. The converse can be proven identically.
\end{proof}

The hierarchy is composed of increasingly larger sets of gates, where $\mathsf{Clifford} \mathcal{}_n \subset \mathsf{Clifford}_{n + 1}$. These sets are closed under group multiplication only for $n \leq 2$. Furthermore, $\mathsf{Clifford}_n /\mathbbm{C}$ are finite sets. For $n > 2$, $\mathsf{Clifford}_n$ generate a dense subset of the full unitary group. A full characterization of subgroups included in $\mathcal{P}_n$ remains an interesting open problem.

\section{Constant depth local circuits (proof of lemma~\ref{lemma:hierarchy2})}\label{sec:proof}

\begin{proof}
Let us assume that the unitary $U$ preserves the codespace and is implementable by a constant depth local quantum circuit with the spread $s_U$. The proof proceeds by induction. Assuming $m = 0$, the dressed logical operator $U$ is supported on a bare-cleanable region and by definition~\ref{def:BareCleanable1} must be a trivial logical operator in $\mathcal{P}_0$.
  
Let us now assume that our statement is true up to $m$ and prove it for $m + 1$. Consider a unitary $U$ with the spread $s_U$ such that
\begin{equation}
    \tmop{supp} ( U) \subseteq \bigcup_{j = 0}^{m + 1} R_j. 
\end{equation}
Any logical Pauli operator $[P]_L$ has a dressed incarnation $P$ fully supported on $\left(R_{1}^+\right)^c$. Observe that
\begin{align}
\tmop{supp} ( UPU^\dagger P^\dagger) & \subseteq \left[  \bigcup_{j = 0}^{m + 1} R_j \right] \cap \mathfrak{B} ( \left(R_{1}^+\right)^c, s_U) \\
    & \subseteq  R_0 \cup \bigcup_{j = 2}^{m + 1} R_j.
\end{align}
Furthermore, we have that $s_{UPU^\dagger P^\dagger} \leq 2 s_U$. Hence, by inductive hypothesis, $ [U{P}U^\dagger {P}^\dagger]_L$, which is also a dressed logical operator, must belong to $\mathcal{P}_m$ when restricted to the codespace. Thus, by definition of the Clifford hierarchy, $[U]_L \in {\mathcal{P}}_{m + 1}$.
\end{proof}


\begin{thebibliography}{43}
\expandafter\ifx\csname natexlab\endcsname\relax\def\natexlab#1{#1}\fi
\expandafter\ifx\csname bibnamefont\endcsname\relax
  \def\bibnamefont#1{#1}\fi
\expandafter\ifx\csname bibfnamefont\endcsname\relax
  \def\bibfnamefont#1{#1}\fi
\expandafter\ifx\csname citenamefont\endcsname\relax
  \def\citenamefont#1{#1}\fi
\expandafter\ifx\csname url\endcsname\relax
  \def\url#1{\texttt{#1}}\fi
\expandafter\ifx\csname urlprefix\endcsname\relax\def\urlprefix{URL }\fi
\providecommand{\bibinfo}[2]{#2}
\providecommand{\eprint}[2][]{\url{#2}}

\bibitem[{\citenamefont{Shor}(1996)}]{Shor96}
\bibinfo{author}{\bibfnamefont{P.~W.} \bibnamefont{Shor}}, in
  \emph{\bibinfo{booktitle}{Proceedings of the 37th Annual Symposium on
  Foundations of Computer Science (FOCS)}} (\bibinfo{publisher}{IEEE Computer
  Society, Los Alamitos, CA}, \bibinfo{year}{1996}), p.~\bibinfo{pages}{56}.

\bibitem[{\citenamefont{Preskill}(1998)}]{Preskill98}
\bibinfo{author}{\bibfnamefont{J.}~\bibnamefont{Preskill}},
  \bibinfo{journal}{Proc. Roy. Soc. Lond.} \textbf{\bibinfo{volume}{454}},
  \bibinfo{pages}{385} (\bibinfo{year}{1998}).

\bibitem[{\citenamefont{Eastin and Knill}(2009)}]{Eastin09}
\bibinfo{author}{\bibfnamefont{B.}~\bibnamefont{Eastin}} \bibnamefont{and}
  \bibinfo{author}{\bibfnamefont{E.}~\bibnamefont{Knill}},
  \bibinfo{journal}{Phys. Rev. Lett.} \textbf{\bibinfo{volume}{102}},
  \bibinfo{pages}{110502} (\bibinfo{year}{2009}).

\bibitem[{\citenamefont{Bravyi and K{\"o}nig}(2013)}]{Bravyi13b}
\bibinfo{author}{\bibfnamefont{S.}~\bibnamefont{Bravyi}} \bibnamefont{and}
  \bibinfo{author}{\bibfnamefont{R.}~\bibnamefont{K{\"o}nig}},
  \bibinfo{journal}{Phys. Rev. Lett.} \textbf{\bibinfo{volume}{110}},
  \bibinfo{pages}{170503} (\bibinfo{year}{2013}).

\bibitem[{\citenamefont{Gottesman and Chuang}(1999)}]{Gottesman99}
\bibinfo{author}{\bibfnamefont{D.}~\bibnamefont{Gottesman}} \bibnamefont{and}
  \bibinfo{author}{\bibfnamefont{I.~L.} \bibnamefont{Chuang}},
  \bibinfo{journal}{Nature} \textbf{\bibinfo{volume}{402}},
  \bibinfo{pages}{390} (\bibinfo{year}{1999}), \eprint{0906.1579v1}.

\bibitem[{\citenamefont{{Nielsen} and {Chuang}}(2000)}]{Nielsen_Chuang}
\bibinfo{author}{\bibfnamefont{M.~A.} \bibnamefont{{Nielsen}}}
  \bibnamefont{and} \bibinfo{author}{\bibfnamefont{I.~L.}
  \bibnamefont{{Chuang}}}, \emph{\bibinfo{title}{Quantum Computation and
  Quantum Information}} (\bibinfo{publisher}{Cambridge University Press,
  Cambridge}, \bibinfo{year}{2000}).

\bibitem[{\citenamefont{Gottesman}(1998)}]{Gottesman98}
\bibinfo{author}{\bibfnamefont{D.}~\bibnamefont{Gottesman}},
  \bibinfo{journal}{Phys. Rev. A} \textbf{\bibinfo{volume}{57}},
  \bibinfo{pages}{127} (\bibinfo{year}{1998}).

\bibitem[{\citenamefont{Bravyi}(2006)}]{Bravyi05}
\bibinfo{author}{\bibfnamefont{S.}~\bibnamefont{Bravyi}},
  \bibinfo{journal}{Phys. Rev. A} \textbf{\bibinfo{volume}{73}},
  \bibinfo{pages}{042313} (\bibinfo{year}{2006}).

\bibitem[{\citenamefont{Bravyi and Terhal}(2009)}]{Bravyi09}
\bibinfo{author}{\bibfnamefont{S.}~\bibnamefont{Bravyi}} \bibnamefont{and}
  \bibinfo{author}{\bibfnamefont{B.}~\bibnamefont{Terhal}},
  \bibinfo{journal}{New. J. Phys.} \textbf{\bibinfo{volume}{11}},
  \bibinfo{pages}{043029} (\bibinfo{year}{2009}).

\bibitem[{\citenamefont{Bombin and Martin-Delgado}(2006)}]{Bombin06}
\bibinfo{author}{\bibfnamefont{H.}~\bibnamefont{Bombin}} \bibnamefont{and}
  \bibinfo{author}{\bibfnamefont{M.~A.} \bibnamefont{Martin-Delgado}},
  \bibinfo{journal}{Phys. Rev. Lett.} \textbf{\bibinfo{volume}{97}},
  \bibinfo{pages}{180501} (\bibinfo{year}{2006}).

\bibitem[{\citenamefont{Bombin}()}]{Bombin14}
\bibinfo{author}{\bibfnamefont{H.}~\bibnamefont{Bombin}},
  \eprint{arXiv:1311.0879}.

\bibitem[{\citenamefont{Bravyi et~al.}(2010)\citenamefont{Bravyi, Poulin, and
  Terhal}}]{Bravyi10}
\bibinfo{author}{\bibfnamefont{S.}~\bibnamefont{Bravyi}},
  \bibinfo{author}{\bibfnamefont{D.}~\bibnamefont{Poulin}}, \bibnamefont{and}
  \bibinfo{author}{\bibfnamefont{B.}~\bibnamefont{Terhal}},
  \bibinfo{journal}{Phys. Rev. Lett.} \textbf{\bibinfo{volume}{104}},
  \bibinfo{pages}{050503} (\bibinfo{year}{2010}).

\bibitem[{\citenamefont{Bravyi}(2011)}]{Bravyi11}
\bibinfo{author}{\bibfnamefont{S.}~\bibnamefont{Bravyi}},
  \bibinfo{journal}{Phys. Rev. A} \textbf{\bibinfo{volume}{83}},
  \bibinfo{pages}{012320} (\bibinfo{year}{2011}).

\bibitem[{\citenamefont{Gottesman}(1996)}]{Gottesman96}
\bibinfo{author}{\bibfnamefont{D.}~\bibnamefont{Gottesman}},
  \bibinfo{journal}{Phys. Rev. A} \textbf{\bibinfo{volume}{54}},
  \bibinfo{pages}{1862} (\bibinfo{year}{1996}).

\bibitem[{\citenamefont{Bombin and Martin-Delgado}(2009)}]{Bombin09c}
\bibinfo{author}{\bibfnamefont{H.}~\bibnamefont{Bombin}} \bibnamefont{and}
  \bibinfo{author}{\bibfnamefont{M.~A.} \bibnamefont{Martin-Delgado}},
  \bibinfo{journal}{Journal of Physics A: Mathematical and General}
  \textbf{\bibinfo{volume}{42}}, \bibinfo{pages}{095302}
  (\bibinfo{year}{2009}).

\bibitem[{\citenamefont{Poulin}(2005)}]{Poulin05}
\bibinfo{author}{\bibfnamefont{D.}~\bibnamefont{Poulin}},
  \bibinfo{journal}{Phys. Rev. Lett.} \textbf{\bibinfo{volume}{95}},
  \bibinfo{pages}{230504} (\bibinfo{year}{2005}).

\bibitem[{\citenamefont{Kribs et~al.}(2006)\citenamefont{Kribs, Laflamme,
  Poulin, and Lesosky}}]{Kribs06}
\bibinfo{author}{\bibfnamefont{D.~W.} \bibnamefont{Kribs}},
  \bibinfo{author}{\bibfnamefont{R.}~\bibnamefont{Laflamme}},
  \bibinfo{author}{\bibfnamefont{D.}~\bibnamefont{Poulin}}, \bibnamefont{and}
  \bibinfo{author}{\bibfnamefont{M.}~\bibnamefont{Lesosky}},
  \bibinfo{journal}{Quant. Inf. Comp.} \textbf{\bibinfo{volume}{6}},
  \bibinfo{pages}{383} (\bibinfo{year}{2006}).

\bibitem[{\citenamefont{Nielsen and Poulin}(2007)}]{Nielsen07}
\bibinfo{author}{\bibfnamefont{M.}~\bibnamefont{Nielsen}} \bibnamefont{and}
  \bibinfo{author}{\bibfnamefont{D.}~\bibnamefont{Poulin}},
  \bibinfo{journal}{Phys. Rev. A} \textbf{\bibinfo{volume}{75}},
  \bibinfo{pages}{064304} (\bibinfo{year}{2007}).

\bibitem[{\citenamefont{Bacon}(2006)}]{Bacon06}
\bibinfo{author}{\bibfnamefont{D.}~\bibnamefont{Bacon}},
  \bibinfo{journal}{Phys. Rev. A} \textbf{\bibinfo{volume}{73}},
  \bibinfo{pages}{012340} (\bibinfo{year}{2006}).

\bibitem[{\citenamefont{Yoshida and Chuang}(2010)}]{Beni10}
\bibinfo{author}{\bibfnamefont{B.}~\bibnamefont{Yoshida}} \bibnamefont{and}
  \bibinfo{author}{\bibfnamefont{I.~L.} \bibnamefont{Chuang}},
  \bibinfo{journal}{Phys. Rev. A} \textbf{\bibinfo{volume}{81}},
  \bibinfo{pages}{052302} (\bibinfo{year}{2010}).

\bibitem[{\citenamefont{Stace et~al.}(2009)\citenamefont{Stace, Barrett, and
  Doherty}}]{Stace09}
\bibinfo{author}{\bibfnamefont{T.~M.} \bibnamefont{Stace}},
  \bibinfo{author}{\bibfnamefont{S.~D.} \bibnamefont{Barrett}},
  \bibnamefont{and} \bibinfo{author}{\bibfnamefont{A.~C.}
  \bibnamefont{Doherty}}, \bibinfo{journal}{Phys. Rev. Lett.}
  \textbf{\bibinfo{volume}{102}} (\bibinfo{year}{2009}).

\bibitem[{\citenamefont{Dorier et~al.}(2005)\citenamefont{Dorier, Becca, and
  Mila}}]{Dorier05}
\bibinfo{author}{\bibfnamefont{J.}~\bibnamefont{Dorier}},
  \bibinfo{author}{\bibfnamefont{F.}~\bibnamefont{Becca}}, \bibnamefont{and}
  \bibinfo{author}{\bibfnamefont{F.}~\bibnamefont{Mila}},
  \bibinfo{journal}{Physical Review B} \textbf{\bibinfo{volume}{72}},
  \bibinfo{pages}{024448} (\bibinfo{year}{2005}).

\bibitem[{\citenamefont{Kovalev and Pryadko}(2013)}]{Kovalev13}
\bibinfo{author}{\bibfnamefont{A.~A.} \bibnamefont{Kovalev}} \bibnamefont{and}
  \bibinfo{author}{\bibfnamefont{L.~P.} \bibnamefont{Pryadko}},
  \bibinfo{journal}{Phys. Rev. A} \textbf{\bibinfo{volume}{87}},
  \bibinfo{pages}{020304} (\bibinfo{year}{2013}).

\bibitem[{\citenamefont{Pastawski et~al.}(2010)\citenamefont{Pastawski, Kay,
  Schuch, and Cirac}}]{Pastawski09}
\bibinfo{author}{\bibfnamefont{F.}~\bibnamefont{Pastawski}},
  \bibinfo{author}{\bibfnamefont{A.}~\bibnamefont{Kay}},
  \bibinfo{author}{\bibfnamefont{N.}~\bibnamefont{Schuch}}, \bibnamefont{and}
  \bibinfo{author}{\bibfnamefont{I.}~\bibnamefont{Cirac}},
  \bibinfo{journal}{Quantum Inf. Comput.} \textbf{\bibinfo{volume}{10}},
  \bibinfo{pages}{580} (\bibinfo{year}{2010}).

\bibitem[{\citenamefont{Hastings}(2011)}]{Hastings11}
\bibinfo{author}{\bibfnamefont{M.~B.} \bibnamefont{Hastings}},
  \bibinfo{journal}{Phys. Rev. Lett.} \textbf{\bibinfo{volume}{107}},
  \bibinfo{pages}{210501} (\bibinfo{year}{2011}).

\bibitem[{\citenamefont{Dennis et~al.}(2002)\citenamefont{Dennis, Kitaev,
  Landahl, and Preskill}}]{Dennis02}
\bibinfo{author}{\bibfnamefont{E.}~\bibnamefont{Dennis}},
  \bibinfo{author}{\bibfnamefont{A.}~\bibnamefont{Kitaev}},
  \bibinfo{author}{\bibfnamefont{A.}~\bibnamefont{Landahl}}, \bibnamefont{and}
  \bibinfo{author}{\bibfnamefont{J.}~\bibnamefont{Preskill}},
  \bibinfo{journal}{J. Math. Phys.} \textbf{\bibinfo{volume}{43}},
  \bibinfo{pages}{4452} (\bibinfo{year}{2002}).

\bibitem[{\citenamefont{Yoshida}(2011)}]{Beni11}
\bibinfo{author}{\bibfnamefont{B.}~\bibnamefont{Yoshida}},
  \bibinfo{journal}{Ann. Phys.} \textbf{\bibinfo{volume}{326}},
  \bibinfo{pages}{2566} (\bibinfo{year}{2011}).

\bibitem[{\citenamefont{Haah and Preskill}()}]{Haah10}
\bibinfo{author}{\bibfnamefont{J.}~\bibnamefont{Haah}} \bibnamefont{and}
  \bibinfo{author}{\bibfnamefont{J.}~\bibnamefont{Preskill}},
  \bibinfo{howpublished}{arXiv:1011.3529}.

\bibitem[{\citenamefont{Landon-Cardinal and Poulin}(2013)}]{Landon-Cardinal13}
\bibinfo{author}{\bibfnamefont{O.}~\bibnamefont{Landon-Cardinal}}
  \bibnamefont{and} \bibinfo{author}{\bibfnamefont{D.}~\bibnamefont{Poulin}},
  \bibinfo{journal}{Phys. Rev. Lett.} \textbf{\bibinfo{volume}{110}},
  \bibinfo{pages}{090502} (\bibinfo{year}{2013}).

\bibitem[{\citenamefont{Haah}(2011)}]{Haah11}
\bibinfo{author}{\bibfnamefont{J.}~\bibnamefont{Haah}}, \bibinfo{journal}{Phys.
  Rev. A} \textbf{\bibinfo{volume}{83}}, \bibinfo{pages}{042330}
  (\bibinfo{year}{2011}).

\bibitem[{\citenamefont{Bravyi and Haah}(2013)}]{Bravyi13}
\bibinfo{author}{\bibfnamefont{S.}~\bibnamefont{Bravyi}} \bibnamefont{and}
  \bibinfo{author}{\bibfnamefont{J.}~\bibnamefont{Haah}},
  \bibinfo{journal}{Phys. Rev. Lett.} \textbf{\bibinfo{volume}{111}},
  \bibinfo{pages}{200501} (\bibinfo{year}{2013}).

\bibitem[{\citenamefont{Kim}()}]{Kim12}
\bibinfo{author}{\bibfnamefont{I.~H.} \bibnamefont{Kim}},
  \bibinfo{howpublished}{arXiv:1202.0052}.

\bibitem[{\citenamefont{Yoshida}(2013)}]{Beni13}
\bibinfo{author}{\bibfnamefont{B.}~\bibnamefont{Yoshida}},
  \bibinfo{journal}{Phys. Rev. B} \textbf{\bibinfo{volume}{88}},
  \bibinfo{pages}{125122} (\bibinfo{year}{2013}).

\bibitem[{\citenamefont{Michnicki}()}]{Michnicki12}
\bibinfo{author}{\bibfnamefont{K.}~\bibnamefont{Michnicki}},
  \bibinfo{howpublished}{arXiv:1208.3496}.

\bibitem[{\citenamefont{Bombin et~al.}(2013)\citenamefont{Bombin, Chhajlany,
  Horodecki, and Martin-Delgado}}]{Bombin13}
\bibinfo{author}{\bibfnamefont{H.}~\bibnamefont{Bombin}},
  \bibinfo{author}{\bibfnamefont{R.~W.} \bibnamefont{Chhajlany}},
  \bibinfo{author}{\bibfnamefont{M.}~\bibnamefont{Horodecki}},
  \bibnamefont{and} \bibinfo{author}{\bibfnamefont{M.~A.}
  \bibnamefont{Martin-Delgado}}, \bibinfo{journal}{New. J. Phys.}
  \textbf{\bibinfo{volume}{15}}, \bibinfo{pages}{055023}
  (\bibinfo{year}{2013}).

\bibitem[{\citenamefont{Pastawski et~al.}(2011)\citenamefont{Pastawski,
  Clemente, and Cirac}}]{Pastawski11}
\bibinfo{author}{\bibfnamefont{F.}~\bibnamefont{Pastawski}},
  \bibinfo{author}{\bibfnamefont{L.}~\bibnamefont{Clemente}}, \bibnamefont{and}
  \bibinfo{author}{\bibfnamefont{J.~I.} \bibnamefont{Cirac}},
  \bibinfo{journal}{Phys. Rev. A} \textbf{\bibinfo{volume}{83}},
  \bibinfo{pages}{012304} (\bibinfo{year}{2011}).

\bibitem[{\citenamefont{Chen et~al.}(2010)\citenamefont{Chen, Gu, and
  Wen}}]{Chen10}
\bibinfo{author}{\bibfnamefont{X.}~\bibnamefont{Chen}},
  \bibinfo{author}{\bibfnamefont{Z.-C.} \bibnamefont{Gu}}, \bibnamefont{and}
  \bibinfo{author}{\bibfnamefont{X.-G.} \bibnamefont{Wen}},
  \bibinfo{journal}{Phys. Rev. B} \textbf{\bibinfo{volume}{82}},
  \bibinfo{pages}{155138} (\bibinfo{year}{2010}).

\bibitem[{\citenamefont{Bravyi and Kitaev}(2005)}]{Bravyi05b}
\bibinfo{author}{\bibfnamefont{S.}~\bibnamefont{Bravyi}} \bibnamefont{and}
  \bibinfo{author}{\bibfnamefont{A.}~\bibnamefont{Kitaev}},
  \bibinfo{journal}{Phys. Rev. A} \textbf{\bibinfo{volume}{71}},
  \bibinfo{pages}{022316} (\bibinfo{year}{2005}).

\bibitem[{\citenamefont{Paetznick and Reichardt}(2013)}]{Paetznick13}
\bibinfo{author}{\bibfnamefont{A.}~\bibnamefont{Paetznick}} \bibnamefont{and}
  \bibinfo{author}{\bibfnamefont{B.~W.} \bibnamefont{Reichardt}},
  \bibinfo{journal}{Phys. Rev. Lett.} \textbf{\bibinfo{volume}{111}},
  \bibinfo{pages}{090505} (\bibinfo{year}{2013}).

\bibitem[{\citenamefont{Jochym-O'Connor and Laflamme}(2014)}]{OConnor14}
\bibinfo{author}{\bibfnamefont{T.}~\bibnamefont{Jochym-O'Connor}}
  \bibnamefont{and} \bibinfo{author}{\bibfnamefont{R.}~\bibnamefont{Laflamme}},
  \bibinfo{journal}{Phys. Rev. Lett.} \textbf{\bibinfo{volume}{112}},
  \bibinfo{pages}{010505} (\bibinfo{year}{2014}).

\bibitem[{\citenamefont{Yoshida}()}]{Beni14}
\bibinfo{author}{\bibfnamefont{B.}~\bibnamefont{Yoshida}},
  \eprint{arXiv:1404.0457}.

\bibitem[{\citenamefont{Bravyi et~al.}(2011)\citenamefont{Bravyi, Leemhuis, and
  Terhal}}]{Bravyi11b}
\bibinfo{author}{\bibfnamefont{S.}~\bibnamefont{Bravyi}},
  \bibinfo{author}{\bibfnamefont{B.}~\bibnamefont{Leemhuis}}, \bibnamefont{and}
  \bibinfo{author}{\bibfnamefont{B.~M.} \bibnamefont{Terhal}},
  \bibinfo{journal}{Ann. Phys.} \textbf{\bibinfo{volume}{326}},
  \bibinfo{pages}{839} (\bibinfo{year}{2011}).

\bibitem[{\citenamefont{Ocko et~al.}(2011)\citenamefont{Ocko, Chen, Zeng,
  Yoshida, Ji, Ruskai, and Chuang}}]{Ocko11}
\bibinfo{author}{\bibfnamefont{S.~A.} \bibnamefont{Ocko}},
  \bibinfo{author}{\bibfnamefont{X.}~\bibnamefont{Chen}},
  \bibinfo{author}{\bibfnamefont{B.}~\bibnamefont{Zeng}},
  \bibinfo{author}{\bibfnamefont{B.}~\bibnamefont{Yoshida}},
  \bibinfo{author}{\bibfnamefont{Z.}~\bibnamefont{Ji}},
  \bibinfo{author}{\bibfnamefont{M.~B.} \bibnamefont{Ruskai}},
  \bibnamefont{and} \bibinfo{author}{\bibfnamefont{I.~L.}
  \bibnamefont{Chuang}}, \bibinfo{journal}{Phys. Rev. Lett.}
  \textbf{\bibinfo{volume}{106}}, \bibinfo{pages}{110501}
  (\bibinfo{year}{2011}).

\end{thebibliography}

\end{document}